\documentclass[12pt]{article}

\usepackage{amsmath, amssymb, amsfonts, amsthm}
\usepackage{graphicx}
\usepackage{natbib}
\usepackage{indentfirst}
\usepackage{times}
\usepackage{bbm}
\usepackage{lineno}
\usepackage{authblk}
\usepackage{bibunits}

\usepackage[colorlinks, linkcolor=black, citecolor=black,filecolor=black, urlcolor=black]{hyperref}

\newtheorem{theorem}{Theorem}
\newtheorem{lemma}[theorem]{Lemma}
\newtheorem{proposition}[theorem]{Proposition}
\newtheorem{corollary}[theorem]{Corollary}

\begin{document}

\title{The inevitability of unconditionally deleterious substitutions during adaptation}

\author[1]{David M.~McCandlish\thanks{To whom correspondence should be addressed. E-mail: \texttt{davidmc@sas.upenn.edu}. }}
\author[2]{Charles L.~Epstein}
\author[1]{Joshua B.~Plotkin}
\affil[1]{Department of Biology, University of Pennsylvania, Philadelphia, PA}
\affil[2]{Department of Mathematics, University of Pennsylvania, Philadelphia, PA}

\date{}

\maketitle

\begin{abstract}  \normalsize 
\noindent 
Studies on the genetics of adaptation typically neglect the possibility
that a deleterious mutation might fix.  Nonetheless, here we show that, in many
regimes, the first substitution is most often deleterious, even when fitness
is expected to increase in the long term.  In particular, we prove that this
phenomenon occurs under weak mutation for any house-of-cards model with an
equilibrium distribution. We find that the same qualitative results hold under Fisher's
geometric model.  We also provide a simple intuition for the surprising
prevalence of unconditionally deleterious substitutions during early adaptation.
Importantly, the phenomenon we describe occurs on fitness landscapes without any local maxima 
and is therefore distinct from ``valley-crossing".
Our results imply that the common practice of ignoring deleterious substitutions
leads to qualitatively incorrect predictions in many regimes. Our results also
have implications for the substitution process at equilibrium and for the
response to a sudden decrease in population size.
\end{abstract}

\newpage

\setlength{\parindent}{.62cm}
\setlength{\parskip}{2ex plus0.5ex minus0.2ex}
\setlength{\abovedisplayskip}{4ex}
\setlength{\belowdisplayskip}{4ex}

\begin{bibunit}[evolution]

\section*{Introduction} Organisms found in nature appear to be exquisitely adapted
to their environments -- so much so that one might think adaptive changes dominate
the process of evolution. However, a more careful examination reveals that
adaptive substitutions cannot be the whole story. First, the genomes of organisms
are filled with elements believed to be deleterious~\citep{Lynch07c}. Second,
genetic drift is known to permit the fixation of both neutral and mildly
deleterious mutations, as is commonly observed in experimental populations
\citep{Halligan09}.  Therefore, a fundamental question for evolutionary biology is
to determine under what conditions a population's fitness will tend to increase as
opposed to decrease.

In this paper we consider this question in a simple case. We suppose that a
population is evolving in the regime of ``weak mutation", so that each new
mutation is either lost or goes to fixation before the next new mutation enters
the
population~\citep[e.g.,][]{Gillespie83,Iwasa88,Sella05,Berg04,McCandlish11,McCandlish13b}.
We also assume the ``house of cards" model, so that the fitness of each new mutant
is drawn independently from a constant
distribution~\citep{Kingman77,Kingman78,Gillespie84landscape,Kauffman87,Ohta90,Flyvbjerg92,Orr02,Jain07,Joyce08,Park08}.
We then ask whether the first mutation that fixes in the population is more likely to
increase or decrease fitness.

The answer to this question is surprising. Even when a population is
destined to adapt towards higher fitness over the long term, the first
mutation to fix will often decrease fitness. We
quantify the effects of the first substitution in two ways. First, we study the
expected selection coefficient of the first substitution. If the expected
selection coefficient is positive, then the fitness of the population is expected
to increase in the short term.  Second, we study the probability that the first
substitution is advantageous. If this probability exceeds one-half, then the
first substitution will be advantageous the majority of the time. Our main result
is a mathematical theorem that characterizes the set of circumstances under which
fitness tends to initially increase or decrease. 

In particular, we show that for essentially any distribution of mutational effects
there exists a range of initial fitnesses such that the expected selection
coefficient of the first substitution is negative. On its own, this result is not
surprising.  After all, if a population starts at the highest
possible fitness then it has nowhere to go but down. What is more surprising is
that this range of initial fitnesses always includes the equilibrium mean
fitness, as well as fitnesses smaller than the equilibrium mean.  In other words,
for many populations undergoing adaptation, fitness is expected to decrease in the
short term even though the expected fitness must eventually increase to its
equilibrium value in the long term.  Likewise, we show that there is a range of initial
fitnesses, including values smaller than the equilibrium median fitness, such that
the first substitution will be deleterious a majority of the time.

Our results on the predominance of downhill steps, even during adaptation, may
sound surprising. But there is a simple, underlying reason why these phenomena
occur.  Even though each individual deleterious mutation is extremely unlikely to
fix, as the population increases in fitness and approaches equilibrium there is an
increasing supply of deleterious mutations~\citep[cf.][]{Hartl98,Silander07}. As a
result, in total there is a substantial chance  -- in fact, often a chance greater
than 50\% -- that one of these deleterious mutations will be the next to fix.  At a
mathematical level, the phenomena we discuss are consequences of a deeper result,
which we also prove, that the fitness achieved after a single substitution is
always probabilistically less than a fitness drawn from the equilibrium
distribution. 

Our results have several counter-intuitive consequences. As already mentioned, they imply that
evolution can be dominated by unconditionally deleterious
substitutions in the short term even if a population is expected to substantially
increase in fitness in the long term. To put this another way, the ``fitness
trajectory''~\citep[i.e.~the expected fitness of the population viewed as a
function of time,][]{Kryazhimskiy09} will often be non-monotonic. There is another
apparently paradoxical consequence of our results: if one begins observing a
well-adapted (i.e.~equilibrial) population at a random time, the next substitution
is more likely to be deleterious than advantageous even though in the
long term the 
frequency of deleterious and advantageous substitutions
must be exactly equal~\citep[see, e.g.][]{Tachida91,Sella05}.

We test the generality of our results by investigating another model commonly used
to study adaptation, Fisher's Geometric Model~\citep{Fisher30,
Kimura84,Hartl96,Hartl98,Orr98,Waxman98,Poon00,Martin06}, which assigns fitnesses based on an
$n$-dimensional continuous phenotype. We find that our results hold in
this case as well, in most regimes. Furthermore we demonstrate the existence of a
``cost of complexity''~\citep{Orr00} where increasing the dimensionality of the
phenotypic space increases the probability that the first substitution is
deleterious.

Our results are important because they challenge two standard ways of thinking
about evolution. First, a very large literature on the genetics of
adaptation~\citep{Orr05} focuses on quantities such as the distribution of
selection coefficients fixed in a sequence of adaptive
substitutions~\citep{Orr98,Orr02,Joyce08}, the number of substitutions that occur
before a population arrives at a local
optimum~\citep{Gillespie84landscape,Jain11}, and the rate of
adaptation~\citep{Orr00,Welch03,Martin06}. Although most studies in this literature
declare by fiat that deleterious fixations cannot occur~\citep[e.g., by using $2s$
as the probability of fixation,][]{Haldane27}, our work shows that a general
theory of adaptation must accommodate deleterious substitutions to achieve 
predictions that are even qualitatively correct.

Second, there is a persistent intuition in the literature that fitness is expected
to increase when a population is below the equilibrium mean fitness and decrease
when a population is above the equilibrium mean fitness.  According to this
intuition, the equilibrium mean fitness is precisely that fitness for which the
expected fitness change is equal to zero. Our results show that this intuition is
false: in fact, fitness is expected to decrease when a population starts at its
equilibrium mean fitness. The standard intuition fails because it erroneously
treats the approach to equilibrium as a deterministic process around the
equilibrium mean, whereas in fact a stochastic treatment is required.

The remainder of our paper is organized as follows. We first explain our
mathematical framework for evolution under weak mutation, and we present our main
results for the house-of-cards model.  We then illustrate our results by
considering a well-studied case where the fitness distribution is
Gaussian~\citep{Tachida91,Tachida96,Gillespie94b}. We follow this by
analyzing Fisher's geometric model, in which we observe and analytically
quantify the same qualitative results found in the house-of-cards model.  We
conclude by discussing the significance of our results in the context of the 
broader literature on adaptation.

\section*{Methods}

\subsection*{Evolutionary dynamics}

We consider a haploid population of size $N$ evolving in the limit of weak
mutation, so that each new mutation either goes to fixation or is lost from the
population before the next new mutation enters the population. In this limiting
regime we can neglect periods of polymorphism and simply model the population as
monomorphic, jumping from one genotype to another at each fixation event. We
assume that new mutations enter the population as a Poisson process with rate $N$.
Thus, time is measured in terms of the expected number of substitutions that would
have accumulated in the population if all mutations were neutral.

Our most important assumption is that the fitness of each new mutation is drawn
independently from some fixed probability distribution $\psi$ that does not depend
on the current fitness of the population. In the literature, this is known as the
House of Cards (HOC) model~\citep{Kingman77,Kingman78}. We assume that $\psi$ is a
continuous probability distribution, and we denote its probability density
function as $\psi(y)$, where $\psi(y)\,dy$ gives the probability that the fitness of
a new mutation lies in the interval $[y,y+dy]$.

Throughout, we assume that fitnesses are measured in terms of relative Malthusian
fitness~\citep[also known as additive fitness,][]{Sella05}, which is the $\log$ of
the relative Wrightian fitness (expected number of offspring divided by the
expected number of offspring of some arbitrary
type)~\citep{Crow70,Wagner10,Houle11}. We define the selection coefficient as the
fitness difference between the new mutant and the allele currently fixed in the
population. This definition approximates the standard selection coefficient when
relative Wrightian fitnesses are close to $1$ (with the approximation becoming
exact in the diffusion limit). While these choices allow a more elegant presentation, 
our results also hold for Wrightian fitnesses and the standard selection coefficient (see~\ref{sec:proofmain}).

Suppose a population is currently fixed for an allele with fitness $x$.  What is
the instantaneous rate of substitution to any other fitness, $y$, which we denote
$Q(y|x)\, dy$? Alleles with fitness $y$ originate within the population by
mutation at rate $N\psi(y)\,dy$, and each such mutation fixes with
probability $u(s,N)$, where $s$ is the selection coefficient, $s=y-x$.
Multiplying the rate of origination by the probability of fixation yields the rate
at which a population jumps from one fitness to another:
\begin{equation}
\label{eq:rates}
\begin{split}
& Q(y|x)\,dy =N\psi(y)\,u(y-x,N)\,dy.
\end{split}
\end{equation}
Thus, if $P(x,t)$ denotes the probability that the population is at fitness $x$ at time
$t$, we have:
\begin{equation}
\label{eq:master}
\frac{\partial P(x,t)}{\partial t}=\int_{-\infty}^{\infty} P(y,t)\,Q(x|y)-P(x,t)\,Q(y|x)\,dy
\end{equation}
In other words, the population's fitness is described by a continuous time and state
Markov process whose transition rates are given by Equation~\ref{eq:rates}. We also assume that at $t=0$ the population has some particular initial fitness, i.e.~$P(x,0)=1$ for some fitness $x$.

In what follows, we use the probability of fixation for a Moran process:~\citep[][see also \citealt{McCandlish13c}]{Moran59}: 
\begin{equation}
u(s,N)=\frac{1-e^{-s}}{1-e^{-Ns}}
\end{equation}
so that our results hold exactly for a haploid Moran process in the limit of weak mutation. 
However, our results also hold approximately for a diploid Wright-Fisher process
in the absence of dominance and in the limit of weak mutation, provided we
adjust appropriately for the difference in chromosomal population size and for the
slight difference in the form of the probability of fixation~\citep[see][]{Sella05}. 

\subsection*{Statistics describing the evolutionary process}

In this section we define some quantities that describe the process of evolution
under weak mutation. We let $k(x)$ denote the substitution rate for a population
with fitness $x$ and $m(x)$ the mean selection coefficient of the first mutation
to fix in a population initially at fitness $x$. Furthermore, we let $p(x)$ denote
the probability that this substitution will be advantageous.

Aside from the first substitution event, we are also interested in how the
expected fitness of the population changes over time.  We let $F_{x}(t)$ denote the
expected fitness of a population at time $t$ given that its fitness was $x$ at
time $t=0$.  Following~\citep{Kryazhimskiy09}, we call $F_{x}(t)$ the fitness
trajectory.

The properties of the first mutation to fix are related to the shape of the
fitness trajectory. In particular, the first derivative of the fitness trajectory
with respect to time at $t=0$ is simply the product of the substitution rate and
the expected selection coefficient of the first substitution: $\left.
\frac{\partial }{\partial t} F_{x}(t)\right|_{t=0}=m(x\,)k(x)$.
However,
as a technical matter it is worth noting that these means,  $m(x)$ and $F_{x}(t)$,
may not necessarily be finite in some instances, for example if $\psi$ has
extremely heavy tails~\citep[see][]{Joyce08}.

\section*{Results}

\subsection*{The equilibrium distribution}

In order to analyze evolution on an HOC landscape we must first determine whether
the population eventually reaches an equilibrium distribution, and, if so, what form
the equilibrium distribution takes.

The equilibrium distribution, which we will write as $\pi$, describes the
long-term probability of finding the population at any given fitness. Moreover,
for a population at equilibrium,  the frequency of
substitutions into fitness class $y$ is equal to the frequency of substitutions
from $y$ to other fitnesses. If an equilibrium distribution $\pi$ exists
it must be unique, since our model operates in continuous time 
and there is a positive
transition rate from every fitness $x$ to the region where $\psi$ is non-zero.

From these conditions, it is easy to show that if $\pi$ exists then its probability density function, $\pi(y)$, must 
be proportional to $\psi(y)e^{(N-1)y}$~\citep{Iwasa88,Berg04,Sella05}. As a
result, the equilibrium $\pi$ exists and its probability density function is given by
\begin{equation}
\label{eq:eq}
\pi(y)=\frac{\psi(y)e^{(N-1)y}}{\int_{-\infty}^{\infty}\psi(z)e^{(N-1)z}\, dz},
\end{equation}
provided that the normalizing constant
$Z=\int_{-\infty}^{\infty}\psi(z)e^{(N-1)z}\, dz$ is finite. This gives a
mathematical condition for the existence of the equilibrium distribution, but not
a biological condition, i.e.~one defined in terms of the evolutionary dynamics.
We have derived such a condition, which will be presented more fully elsewhere. Roughly speaking, this condition states that an equilibrium distribution exists if the fraction 
of
advantageous substitutions, $p(x)$, decreases to zero as the initial fitness, $x$, increases.
More precisely, if the limit as $x\rightarrow \infty$ of $p(x)$ equals 
zero for a population of size $N$, then an equilibrium distribution with finite mean
exists for a population of size $N-1$.

\subsection*{Deleterious substitutions can dominate short-term evolution}

We are now in a position to state our main results. In~\ref{sec:proofmain}, we
prove the following theorem: for any choice of population size $N>1$ and fitness
distribution $\psi$ such that an
equilibrium distribution $\pi$ with a finite mean exists, there exists some
fitness $x_{1}$ such that $m(x)<0$ for any $x>x_{1}$ and furthermore $x_{1}$ is
less than the mean of $\pi$. What this result says is that, for any population
whose starting fitness exceeds some constant, $x_1$, the first substitution to
fix is expected to have a negative selection coefficient. Indeed, the
average selection coefficient of the first substitution is guaranteed to be
negative when starting at the equilibrium mean fitness, and also when starting
with a fitness within some range strictly less than the equilibrium mean.

This result has important consequences for the shape of the fitness trajectory,
$F_{x}(t)$. Recall that $\left. \frac{\partial }{\partial t}
F_{x}(t)\right|_{t=0}$, the initial slope of the fitness trajectory starting at
$x$, has the same sign as mean selection coefficient of the first mutation to fix,
$m(x)$.  Therefore, for any starting fitness in the interval between $x_1$ and the
mean equilibrium fitness, the fitness trajectory must initially be decreasing.
Nonetheless, because asymptotically the fitness trajectory must approach the
equilibrium mean fitness, such trajectories must eventually increase back towards
the equilibrium mean.  Thus, provided an equilibrium with finite mean exists, for
any choice of population size $N>1$ and HOC model $\psi$, there is a range of
starting fitnesses that produce non-monotonic fitness trajectories: fitness is
expected to decrease in the short term, and then increase towards the equilibrium
mean in the long term.

\begin{figure}[p!]
\hspace{-12mm}\includegraphics{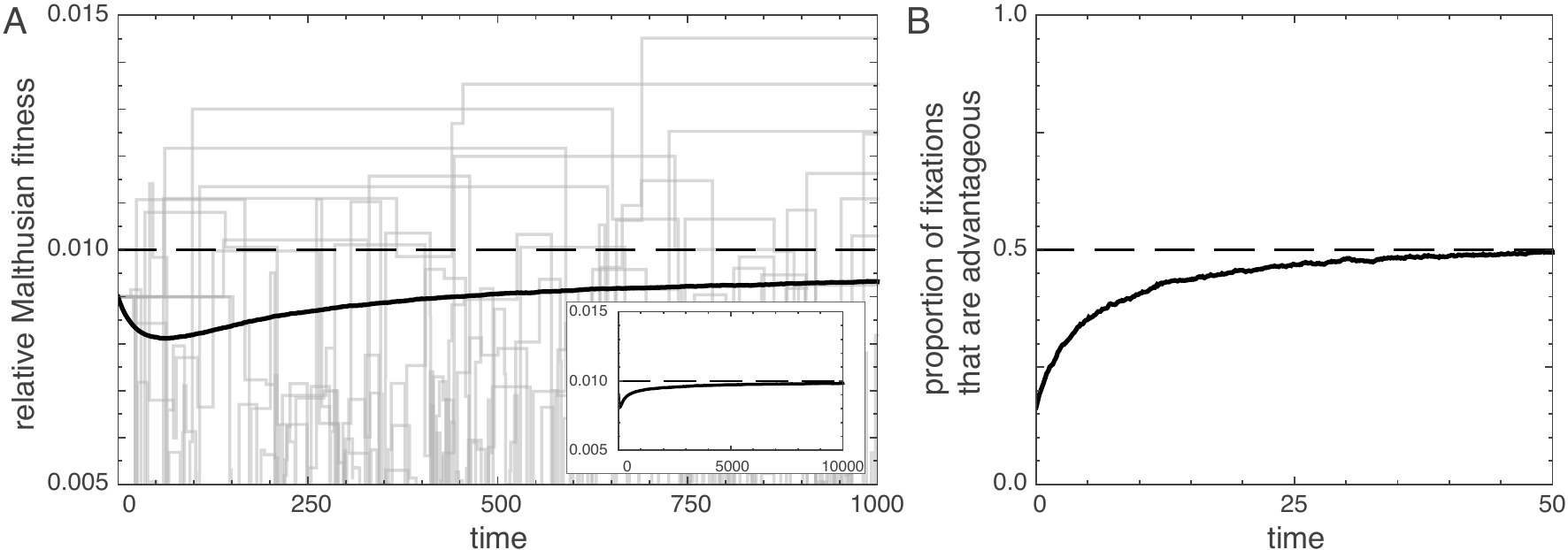}
\caption[]{
A. A non-monotonic fitness trajectory
(black line). The population begins with fitness $.009$, which is below the
equilibrium mean fitness of $.010$ (dashed line). The fitness trajectory (that is,
the ensemble mean fitness) decreases
initially, even though it then increases towards the equilibrium mean in the
long term (inset). Ten individual population realizations are shown in light gray 
(some of these trajectories extend below the region shown on the $y$-axis). B. The
expected proportion of
substitutions that are advantageous, as a function of time, starting again from fitness
$.009$.
Both panels are the result of $100,000$
simulations with the following parameters: $N=1001$, and the mutational distribution
$\psi$ is normal with mean $0$ and variance $10^{-5}$. The proportion of advantageous
fixations was calculated as the expected advantageous substitution
rate across this ensemble of populations, divided by the expected total
substitution rate.
}
\label{fig:trajectories}
\end{figure}

Figure~\ref{fig:trajectories}A gives an example of such a trajectory (black
curve). The population starts below the equilibrium mean fitness (dashed line),
and its expected fitness decreases initially, even though in the long term this
expectation increases to the equilibrium mean fitness (Figure~\ref{fig:trajectories}A, inset). 

Figure~\ref{fig:trajectories}A also shows several realizations of individual
population histories (gray lines), which provide some insight into why the fitness
trajectory (i.e.~the ensemble mean fitness) has a non-monotonic shape. There are
two things to notice about these trajectories. First, the individual trajectories
tend to exhibit early deleterious substitutions. The preponderance of early
deleterious substitutions is shown in Figure~\ref{fig:trajectories}B. Second, notice that the
more fit a population, the longer the waiting time until the next substitution
(see, e.g.~the upper-right corner of Figure~\ref{fig:trajectories}A). This means
that once a population fixes an extremely fit genotype it will tend to remain at
that genotype for a very long time.   At short time scales, however, such
advantageous mutations are unlikely to have occurred, whereas over the long term
it becomes increasingly likely that a high-fitness genotype will enter a
population and go to fixation. Thus, following its short-term decline, the
ensemble mean fitness eventually increases, as populations in the ensemble
eventually acquire a substitution to a very high fitness.

Our first result described above pertains to the expected selection coefficient
of the first substitution, and not to the probability that this first substitution
will be deleterious or advantageous. We have thus derived a corresponding
result that characterizes how often the first substitution is deleterious. In
particular, in~\ref{sec:proofmain}, we prove that for any choice of $N>1$ and
$\psi$ such that an equilibrium distribution $\pi$ exists, there exists some
fitness $x_{2}$ such that $p(x)<1/2$ for $x>x_{2}$, and, furthermore, $x_{2}$ is
less than the median of $\pi$. What this result says is that for any population
whose starting fitness exceeds some constant, $x_2$, the first substitution to fix
is more likely to be deleterious than advantageous. Moreover, 
even if the initial fitness is drawn at random from the equilibrium
distribution then the next substitution will be deleterious a majority of the
time (see~\ref{sec:proofmain}).

\begin{figure}[p!]
\center
\includegraphics{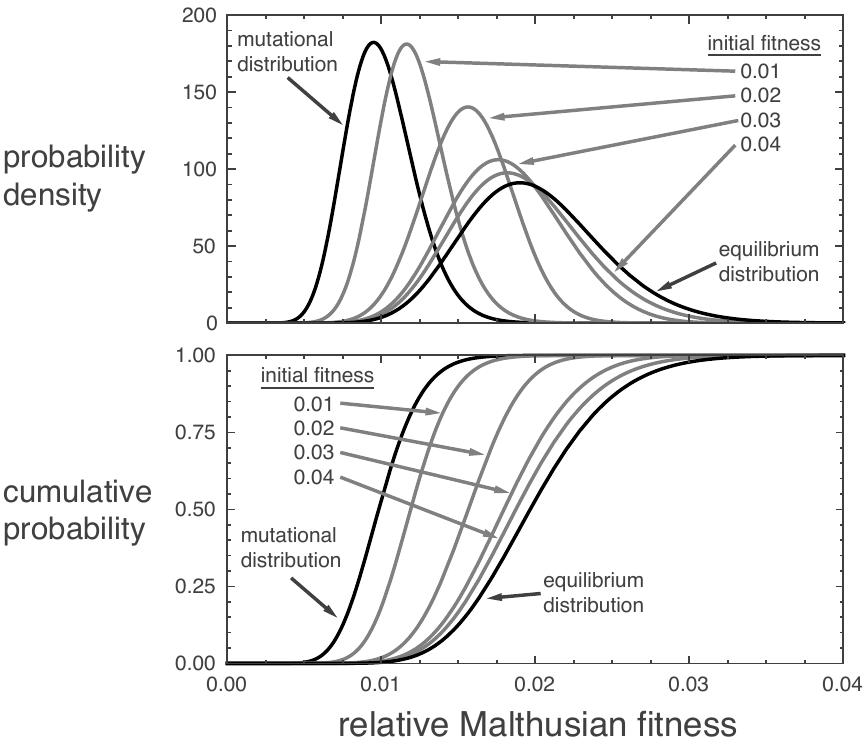}
\caption[HOC example]{
The distribution of
fitnesses following one substitution is bounded between the mutational distribution,
$\psi$, and the
equilibrium distribution, $\pi$. The top panel shows probability density
functions and
the bottom panel shows corresponding cumulative distributions.
Black curves correspond to the mutational
distribution, $\psi$, which is a Gamma-distributed with shape $20$ and mean
$.01$, in this example; and the corresponding equilibrium distribution, $\pi$, which is
Gamma-distributed with
shape $20$ and mean $.02$; $N=1001$. The distribution of fitnesses following the first
substitution are shown in gray lines, for populations initially at
fitness $x=.01$, $x=.02$, $x=.03$ and $x=.04$. Notice
that the cumulative distribution functions for all $x$ are bounded between the
cumulative distribution functions for the mutational distribution and for the
equilibrium distribution; also note that the cumulative distributions shift
monotonically to the right with increasing $x$.}
\label{fig:HOCcdf}
\end{figure}

Our results on the mean and the median fitness after one substitution are both
consequences of a more general set of results we prove in~\ref{sec:proofmain} and~\ref{sec:additional}. These results
yield a simple picture for how the distribution of fitnesses following a
substitution depends on a population's initial fitness.  In particular, we show
that increasing the initial fitness of a population always shifts the distribution
of fitnesses after one substitution to the right, and that this distribution
approaches the equilibrium distribution as the initial fitness increases to
infinity. Likewise, decreasing the initial fitness shifts the distribution of
fitnesses after the first substitution to the left, and this distribution
approaches $\psi$ in the limit of a large, negative initial fitness.  When we say that
a random variable $X$ is to the ``left" of a random variable $Y$ we mean that 
every quantile
of $Y$ is greater than the corresponding quantile of $X$, except perhaps for the
$0$-th and $1$-st quantiles; one might also say that the random variable $X$ is
stochastically less than the random variable $Y$. These results are illustrated
graphically in Figure~\ref{fig:HOCcdf}, which shows the probability distributions
$\pi$ and $\psi$, together with the fitness distribution after the first
substitution, for several different choices of the population's initial
fitness.

The key implication of Figure~\ref{fig:HOCcdf} is that the distribution of
fitnesses following one substitution is always to the left of the equilibrium
distribution, irrespective of the initial fitness. As a result, the mean of the
fitness distribution after one substitution must always be less than the mean of
the equilibrium distribution~(\ref{sec:proofmain}). And so, if a population starts at the equilibrium
mean fitness, then its mean fitness must be reduced by the first substitution
-- which is equivalent to saying that the expected selection coefficient of the
first substitution is negative.  A similar result holds for a population starting
at the median of the equilibrium fitness distribution~(\ref{sec:proofmain}).

In summary, we have shown that under the house of cards model deleterious
substitutions are expected to occur while a population is adapting -- that is,
while a population is still below its equilibrium mean fitness. Indeed,
deleterious substitutions can be more likely to occur than advantageous
substitutions during adaptation.  Moreover, such mutations are unconditionally
deleterious in the sense that they have no productive value for potentiating
subsequent adaptation.  We stress the generality of these results, which hold for
any choice of mutational distribution $\psi$ and for any population size $N>1$ so
long as an equilibrium distribution with finite mean exists. As we shall soon see,
some choices of $\psi$ guarantee such an equilibrium for all $N$, which implies
that the predominance of deleterious substitutions persists even when selection
against deleterious substitutions is arbitrarily strong.

\subsection*{Case study: the Gaussian House of Cards}

Our main result guarantees a range of initial fitnesses for which 
the initial step in adaptation is dominated by deleterious fixations.
But how large is this range of initial fitnesses? 
To investigate this question we consider the best-studied version of the
HOC model, in which the distribution of mutational effects, $\psi$, is Gaussian with
mean $\mu$ and standard deviation $\sigma$. In this case, it has been shown
\citep{Tachida91,Tachida96} that, for any choice of $N$ and $\sigma$, the equilibrium distribution, $\pi$, is also
Gaussian, with standard deviation $\sigma$ and mean $\mu+(N-1)\sigma^{2}$. 

Because $\psi$ and $\pi$ are both normally distributed with the same variance, we
can conduct a nice analysis by exploiting symmetry.
In particular, consider the point half-way between the means of $\psi$
and $\pi$, which we denote by $x^{*}=\mu+(N-1)\sigma^2/2$.
It turns out that the distribution of fitnesses fixed by the first substitution
for a population that starts with fitness $x^{*}+c$ is the same distribution as for a population starting at $x^{*}-c$ when
this distribution is reflected across $x^{*}$ (see \ref{sec:GuassianProofs}).
Thus, $m(x^{*}+c)=-m(x^{*}-c)$ and $p(x^{*}+c)=1-p(x^{*}-c)$. In particular,
this means that $m(x^{*})=0$ and $p(x^{*})=1/2$. 

Intuitively, these results suggest that the region where deleterious substitutions
dominate evolution includes all initial fitness greater than $x^{*}$.
We verified this conjecture using a systematic numerical search for all
parameters $\sigma$ between $10^{-5}$ and $10$ and $(N-1)\sigma$ between $10^{-3}$
and $20$. In all cases examined we found $m(x)<0$ and $p(x)<1/2$ for $x>x^{*}$.
Thus, when a population starts at the mean of the mutational
fitness distribution, deleterious substitutions begin to dominate once the
population's fitness
has increased half-way to its long-term expected value. In other words, there is
a substantial range of fitnesses for which
deleterious substitutions dominate adaptation.

The symmetry argument above also provides some insight about the size of
the selection coefficients of the first substitution. For instance, if the
population starts at the mean fitness of the mutational distribution, then one
intuitively expects the first substitution to have a large, positive effect due to
the abundant supply of advantageous mutations. This intuition is indeed correct.
At the same time, by symmetry, this intuition also implies that if a population starts at
the equilibrium mean fitness, then the first substitution will typically have a large,
deleterious effect. 

\begin{figure}[p!]
\center
\includegraphics{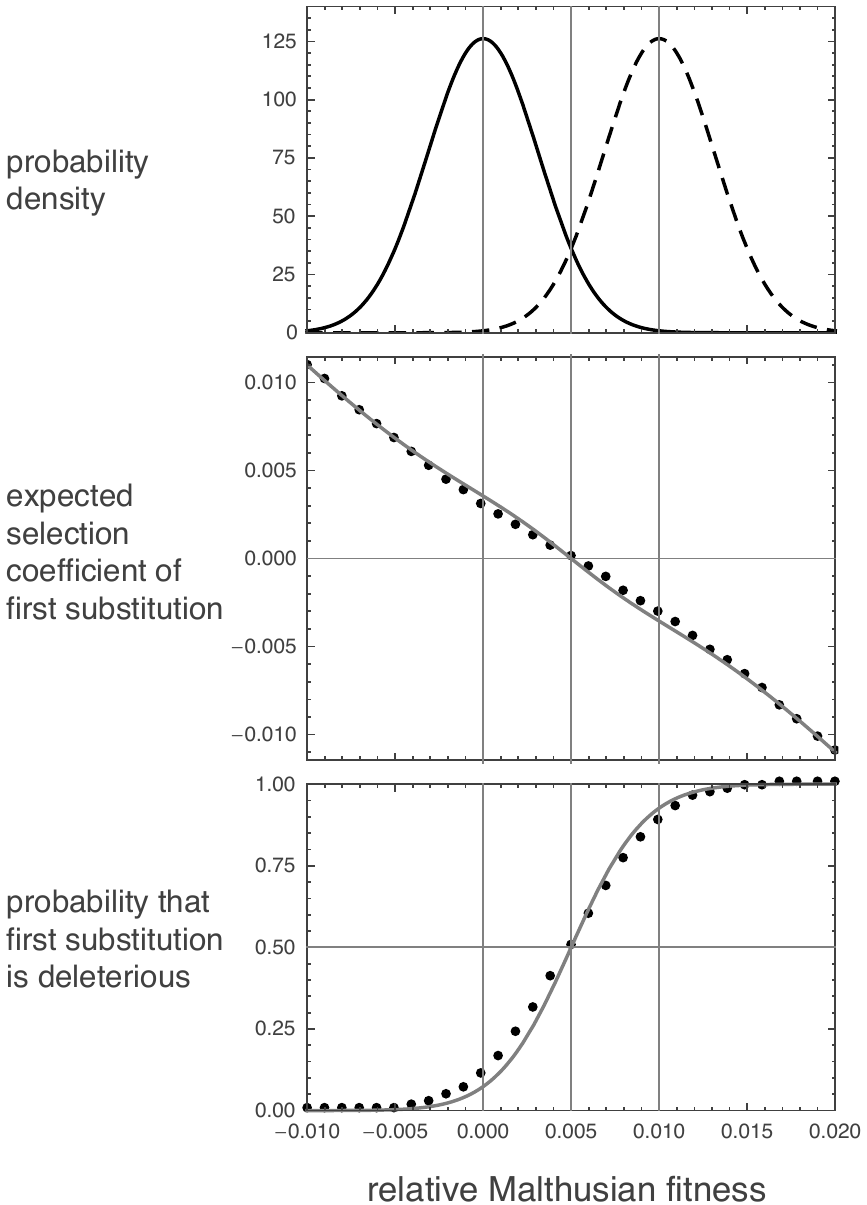}
\caption[HOC example]{
Probability density function of the HOC fitness distribution, $\psi$, and the equilibrium
distribution, $\pi$, together with the expected fitness effect of the first substitution, $m(x)$, and the probability that the first substitution will 
be deleterious, $1-p(x)$, all expressed as functions of the initial fitness, $x$. The HOC distribution is Gaussian with $\mu=0$ and $\sigma^{2}=10^{-5}$; $N=1001$. 
Vertical lines denote the mean of the HOC distribution, the mean fitness at equilibrium, and the average of these two fitnesses, $x^{*}$. Solid lines in the
bottom two plots show our analytical approximations (see~\ref{sec:analytical}); the individual points denote exact results as determined by numerical integration.}
\label{fig:HOCstats}
\end{figure}

In order to illustrate the preceding results, consider the case in which $\mu=0$,
$\sigma^{2}=10^{-5}$ and $N=1001$ (Figure~\ref{fig:HOCstats}). In this case, the
equilibrium fitness distribution is Gaussian with mean $.01$ and variance
$10^{-5}$, so that the average selection coefficient of new mutations, for a
population with the equilibrium mean fitness, is $Ns\approx-10$. Despite this
large, negative average selection coefficient, for a population starting at the
equilibrium mean fitness, the first substitution is deleterious $89\%$ of the
time, and the average selection coefficient of the first substitution is
$Ns=-3.1$.  The bottom two panels of Figure~\ref{fig:HOCstats} also show our
analytical approximation (solid lines) as compared with the exact results (points,
as determined by numerical integration), see~\ref{sec:analytical}.

It is important to remember that the symmetry about $x^{*}=\mu+(N-1)\sigma^2/2$ with
respect to the distribution of fitnesses in the first substitution does not imply
that all dynamics are symmetrical. Note in particular that because $u(y-x,N)$ is increasing in $x$ for $N>1$, the substitution rate, $k(x)$, is a strictly
decreasing function of the fitness $x$ for any HOC model with $N>1$. This implies, for instance, that while the fraction of deleterious substitutions
at $x^{*}+c$ equals the fraction of advantageous substitutions at $x^{*}-c$, the
actual rate of deleterious substitutions at $x^{*}+c$ must be less than the rate
of advantageous substitutions at $x^{*}-c$. In the example above, with $\mu=0$,
$\sigma=.01$ and $N=1001$, the substitution rate at fitness $\mu$ is $1.62$ times the
neutral substitution rate, at fitness $z$ it is $22\%$ of the neutral
substitution, and at the equilibrium mean fitness the substitution rate is $1.1\%$ of the neutral
rate. Thus, while the average fitness effect of a substitution
starting from the equilibrium mean fitness is relatively large and negative, the expected
waiting time for this first fixation to occur is much longer than the waiting time
starting from the mean of the HOC fitness distribution.

\subsection*{Case study: Fisher's geometric model}

Our results for the HOC model do not necessarily hold when the distribution of
fitnesses introduced by mutation is allowed to depend on the current
fitness~\citep{Kryazhimskiy09}. For instance, in this more general class of
correlated fitness landscapes it is possible to find circumstances where the
expected selection coefficient is positive when a population starts at its
equilibrium mean fitness. An immediate question, then, is whether the HOC model is
pathological in some sense, or whether our qualitative results hold for other
commonly used fitness landscapes.

In order to investigate this question, we turn to Fisher's geometric model (FGM),
which has emerged as an important framework for understanding both
adaptive and nearly neutral
evolution~\citep{Fisher30,Kimura85,Hartl96,Hartl98,Orr98}. In addition to its
prominence in the contemporary literature, we have chosen to study FGM because
 it is typically thought of as a paradigmatic example of a smooth,
correlated fitness landscape. This contrasts with the HOC model, which is 
an uncorrelated or rugged
landscape~\citep{Kauffman87}. If our results were caused by
the uncorrelated nature of the HOC model, then we would not expect to find the same
results in Fisher's geometric model.

For the sake of concreteness, we will consider a specific, widely used version of
FGM~\citep{Martin06,Martin07b,Martin08,Tenaillon07,Lourenco11} which assumes that
Malthusian fitness falls off quadratically with the distance to some optimum phenotype in an $n$-dimensional phenotypic space. In particular, we assume that the relative
Malthusian fitness of any phenotype $z$ is given by $-(1/2)||z||^{2}$, where
$||z||$ is the Euclidean distance to the optimum phenotype. This is equivalent to assuming that relative Wrightean fitness is a Gaussian function of the distance to the phenotypic optimum.
In addition, we assume that the distribution of phenotypes produced
by mutation is multivariate Gaussian with variance $\lambda$ centered at the
current phenotype $z$, and that new mutations enter the population as a Poisson
process at rate $N$.

Under this model, for a
population fixed for a phenotype with fitness $x$, the fitness of new mutants is distributed as
$-\lambda/2$ times a non-central chi-squared distributed random variable with $n$
degrees of freedom and non-centrality parameter $-2x/\lambda$~\citep[this follows from Appendix~2 of][]{Martin06}. Furthermore, it can
readily be confirmed using~\citet{Tenaillon07}'s method that the equilibrium
distribution for $N>1$ is $-1$ times a gamma distributed random variable, where
the gamma distribution has shape $n/2$ and scale $1/(N-1)$~\citep[note that this
distribution is independent of $\lambda$,][]{Sella05,Tenaillon07}. In particular,
the equilibrium mean fitness is $-n/(2(N-1))\approx -n/(2N)$, and while no analytical expression exists for the median of this distribution, the median is always greater than the mean~\citep{Chen86}.

The analysis of this model can be simplified by recognizing that the evolutionary dynamics of the model are much easier to understand if we work in scaled fitnesses, $Nx$, instead of fitness. This is because, to a very close approximation, under FGM the evolutionary dynamics of a population when measured in terms of scaled fitnesses depend only on the dimensionality, $n$, and on the compound parameter $N\lambda$ (\ref{sec:FGMscaling}). Thus, in addition to the probability that the first substitution is deleterious when starting from the equilibrium median fitness, it is most useful to examine the expected scaled selection coefficient ($Ns$) for a population starting at the equilibrium mean fitness and to consider the behavior of these quantities as a function of $n$ and $N\lambda$.

\begin{figure}[p!]
\center
\hspace{1cm}
\includegraphics{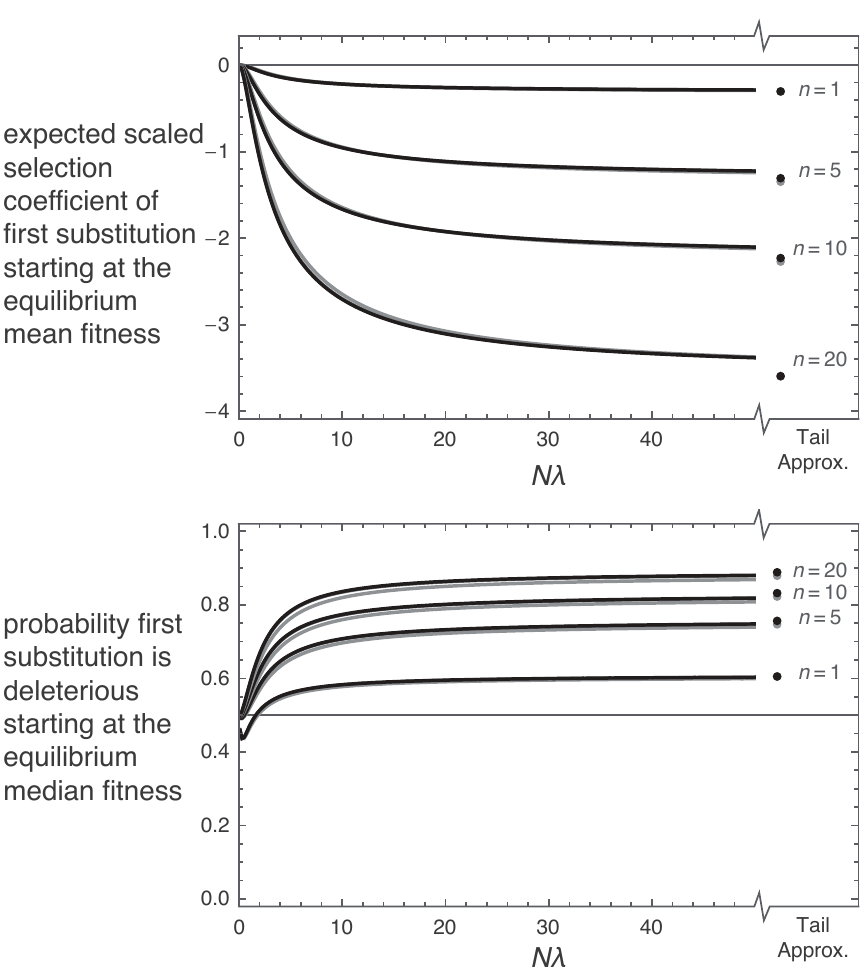}
\caption[FGM example]{
Expected scaled selection coefficient of the first substitution, starting from the equilibrium
mean fitness, and fraction of substitutions that are deleterious, starting from the equilibrium median
fitness, under Fisher's geometric model. Both panels show these quantities as a
function of $N\lambda$, for three values of $N$ ($N=11$, light gray; $N=101$, dark
gray; $N=1001$, black) and for several different values of the phenotypic
dimensionality $n$.
The lines for $N=101$ and $N=1001$ are almost completely overlapping. Small
circles show the asymptotic, large $N\lambda$ values, based on~\citet{Martin08}'s tail
approximation.}
\label{fig:FGMstats}
\end{figure}

Figure~\ref{fig:FGMstats} shows us that the evolutionary
dynamics under FGM fall in one of two regimes, depending on the value of
$N\lambda$.  When $N\lambda$ is very close to zero, the expected scaled
selection coefficient of the first substitution is also close to zero (or indeed,
sometimes very slightly positive). Similarly, in this regime, the fraction of
substitutions that are deleterious is approximately one-half. Thus, if mutations
have extremely small effects, the naive intuitions discussed in the Introduction
are basically accurate.

However, when $N\lambda$ is greater than approximately 2, the situation is quite different: all of the qualitative results we derived for the
HOC model also hold in this regime. In particular, the expected selection coefficient starting from the equilibrium mean fitness is negative; and the
majority of substitutions starting from the equilibrium median fitness are deleterious (Figure~\ref{fig:FGMstats}).  Furthermore, in this regime
increasing the dimensionality of phenotype space decreases the expected scaled selection coefficient of the first substitution and increases the proportion deleterious first substitutions~(Figure~\ref{fig:FGMstats}).  This constitutes a ``cost of
complexity''~\citep{Orr00} in that increasing the dimensionality of the phenotypic
space makes the dynamics less favorable for an adapting population.

To better understand this large $N\lambda$ regime, it is helpful to note that both of our statistics of interest become less sensitive to $N\lambda$
as $N\lambda$ becomes large. In this regime, most mutations are deleterious and strongly selected against, which suggests that mutations destined for
fixation will tend to come from the right tail of the mutational fitness distribution. Indeed, as $N\lambda$ becomes large, one can show that the mutational distribution converges to a tail
approximation previously introduced by~\citet{Martin08} (see \ref{sec:largelim}).  This tail approximation also suggests a method for
deriving analytical results in the large $N\lambda$ regime.  In particular,~\citet{Martin08} derive a number of results using this tail approximation
in conjunction with an approximation for the probability of fixation that allows only for advantageous fixations. By modifying their method to allow
for the possibility of deleterious fixations, we obtain the following simple expression for the proportion of substitutions that will be advantageous,
starting from the equilibrium median fitness (\ref{sec:FGManalytic}): \begin{equation} \frac{1}{2+n/2}.  \end{equation} This expression clearly
indicates that, in the large $N\lambda$ regime, the proportion of deleterious substitutions from the equilibrium median fitness exceeds $1/2$ and is an increasing
function of the phenotypic dimensionality, $n$.  Applying similar methods to the expected scaled selection coefficient of the first substitution from
the equilibrium mean fitness indicates that for large $N\lambda$ this quantity is 1) always negative and 2) solely a function of  $n$  (\ref{sec:FGManalytic}).

\section*{Discussion}

We have demonstrated that evolution can be dominated by deleterious substitutions
in the short term even when adaptation will eventually occur over the long term.
Critically, this initial decrease in fitness is not due to alleles that potentiate
or otherwise promote the subsequent increase in fitness~\citep{Cowperthwaite06},
as would occur if these populations were crossing fitness valleys~\citep{Wright32,
Kimura85,vanNimwegen00,Weinreich05a,Weissman09}.  Rather, this decrease in fitness
is due to unconditionally deleterious substitutions that have absolutely no effect
other than to decrease the fitness of the population. Our results imply that
studies of adaptation that ignore deleterious substitutions must give qualitatively
incorrect predictions, at least in some regimes.

More precisely, we have considered a population evolving under a Moran process in
the limit of weak mutation, where the fitnesses of new mutations are drawn
independently from some fixed distribution~\citep[i.e., the house of cards
model]{Kingman77,Kingman78}. Our results hold for arbitrarily large populations,
and hence arbitrarily strong selection, under only the condition that a long-term
equilibrium distribution exists with finite mean.  This condition is always
satisfied under the biologically plausible assumption that the probability that
the next substitution be advantageous approaches zero as the fitness of a
population increases to infinity.

Our two main results guarantee a region of initial fitnesses, below the equilibrium mean fitness, where the expected selection coefficient of the
first substitution is negative; and a region of initial fitnesses, below the equilibrium median fitness, where the first substitution will be
deleterious a majority of the time. In addition, we have shown qualitatively similar behavior under Fisher's Geometric Model, where a ``cost of
complexity''~\citep{Orr00,Martin06} also arises, so that the magnitude of these effects becomes larger as the dimensionality of the phenotypic space
increases~\citep[c.f.][]{Fisher30}. 

Our results are surprising, and so it is important to develop some intuitions for
why they hold. At the most basic level, our results are possible because the
process of adaptation itself tends to increase fitness until the vast majority of
possible mutations are deleterious~\citep{Hartl98,Silander07}. Indeed, this
majority is eventually so vast that, at some sufficiently high fitness,
deleterious mutations become responsible for the bulk of substitutions. 

At a slightly more detailed level, the question is one of time scales: natural
selection is so efficient that in the long term a population will spend most of
its time at extremely rare, high-fitness alleles. However, the waiting time for
such alleles to enter the population by mutation is extremely long, so that many
populations will experience a deleterious substitution before even having the
opportunity to fix an allele with high fitness. Thus, fitness may tend to decrease
in the short term even though it will on average increase in the long term. 

Finally, at a more mathematical level, our results follow from a deeper theorem
that states that, no matter what the initial fitness, every quantile of the
distribution of fitnesses after one substitution is less than the corresponding
quantile of the equilibrium fitness distribution.  We have used this deeper result to
derive our two main results, but it has other consequences that we have not yet
investigated in detail.  For instance, consider a population that has experienced
a recent change in environment or in population size, such that its current
fitness is much higher than the equilibrium distribution. In this case, our result
implies that the distribution of fitnesses after one substitution occurs will be
to the left of the equilibrium distribution.  Thus, fitness does not decrease to
the new equilibrium by many small steps, but rather takes one large downhill jump.

The relationship between the equilibrium fitness distribution and the distribution of fitnesses after one substitution has an intuitive basis. 
In essence, high-fitness genotypes are occupied at high frequency in equilibrium for two distinct reasons:
\begin{enumerate}
\item A population is more likely to become fixed for a high-fitness allele because the probability of fixation is greater for high-fitness mutations.
\item Once fixed at a high-fitness allele, the time until the next fixation tends to be long, because of selection against deleterious substitutions.
\end{enumerate}
While both of these factors push the equilibrium distribution towards very high fitnesses, 
only the first factor influences the distribution of fitnesses after one substitution. As a result, the distribution of fitnesses at equilibrium is 
always stochastically greater than the distribution of fitnesses after one substitution.

This distinction between the two factors above, and their different effects on
evolution, is often neglected in the literature.  In particular, the weak-mutation literature is
primarily composed of two types of Markov models.  In the first type of model,
time is measured in units that are independent of when substitution events occur
-- so that time continues to elapse during the periods between subsequent
substitution events in the population~\citep[see e.g.][]{Iwasa88,Sella05,Berg04,
Kryazhimskiy09,McCandlish13b}.  In the second type of model, time is always
discrete, and each unit of time corresponds to a single substitution event, no
matter how long a population actually spends waiting in between substitution
events
~\citep[e.g.~][]{Gillespie83,Gillespie84landscape,Hartl98,Orr02,Weinreich06,
Draghi13}. Mathematically, a model of the first type is called the ``full chain'';
the corresponding model of the second type is called the ``embedded
chain''. Both of these Markov processes are important theoretical objects.
However, while both of the above-mentioned factors influence the behavior of the
full chain, only the first factor influences the behavior of the corresponding
embedded chain. This means that the full chain and the embedded chain can have very different properties.

For instance, at equilibrium, evolution under the house of cards model is
characterized by periods of rapid substitutions among relatively low-fitness
alleles interspersed with long periods of stasis at high-fitness
alleles~\citep{Iwasa93}. As a result, the equilibrium distribution of the full
chain is concentrated at higher fitnesses than the equilibrium distribution of the
embedded chain, since the embedded chain neglects the extra amount of time a
population typically spends fixed at high-fitness alleles (more precisely, the probability density function of the equilibrium distribution for the embedded chain is proportional to $\pi(x)\,k(x)$, where the substitution rate $k(x)$ is a decreasing function of $x$). Which type of these two
equilibrium distributions is relevant for a particular purpose depends on whether
one observes a population at a random time (as is most often the case for
empirical data), in which case the full chain is appropriate; or whether one
observes a population at a random substitution, in which case the embedded
chain is appropriate. 

Our results also have several important implications for the common practice of
ignoring deleterious substitutions when studying adaptation. First, by
demonstrating that adaptation can occur even when unconditionally deleterious
substitutions dominate short-term evolution, we have shown that any satisfactory,
general theory of adaptation must allow for the possibility of deleterious
substitutions.

Second, our results characterize the conditions under which deleterious
substitutions initially {\it dominate} the the adaptive process, i.e.~when a
majority of substitutions are deleterious or when fitness decreases in
expectation.  The range of parameters where deleterious substitutions play a
non-negligible role is likely to be much larger.

Third, we emphasize that the frequency of deleterious substitutions changes as a
population evolves. This frequency is therefore a property of the evolutionary
dynamics rather than a parameter that can be assigned at the outset.

Finally, the problem of neglecting deleterious substitutions is particularly acute
for studies that depend on ``extreme value theory'', or ``tail approximations,''
because the mathematical methods underlying such studies assume the current
fitness is already in the extreme right tail of the mutational fitness
distribution. This is precisely the regime where we expect deleterious
substitutions to be most important and therefore precisely the regime where they
should not be neglected. One way forward for these studies is to
incorporate approximations that allow deleterious substitutions to occur, as we
did here in our analysis of Fisher's Geometric Model. 

An important limitation of our current analysis is that it has been conducted in the limit of weak mutation. It thus remains an open question how
broadly these phenomena occur in other population-genetic regimes. However, a recent study of cancer progression by~\citet{McFarland13} suggests that
similar dynamics, featuring a predominance of deleterious substitutions despite long-term adaptation, may occur even in populations with substantial polymorphism.

\section*{Acknowledgements} We thank Ricky Der and Mitchell Johnson for fruitful discussions and Warren Ewens for comments on the manuscript. J.B.P. acknowledges funding from the Burroughs Wellcome Fund, the 
David and
Lucile Packard Foundation, the James S.~McDonnell Foundation, the Alfred P.~Sloan Foundation, the U.S.~Department of the
Interior (D12AP00025), and the Foundational Questions in Evolutionary Biology Fund (RFP-12-16). J.B.P., C.L.E. and D.M.M. acknowledge funding
from the U.S.~Army Research Office (W911NF-12-1-0552). 

\clearpage
\pagebreak


\renewcommand{\refname}{Literature Cited}
\putbib[MainBibtexDatabase]

\clearpage
\pagebreak

\def\thesection{Appendix \Alph{section}}
\renewcommand{\theequation}{A\arabic{equation}}
\setcounter{equation}{0}

\section{Proof of the main result}
\label{sec:proofmain}
In this appendix, we provide a characterization of the distribution of fitnesses after one substitution. In particular, we provide bounds on the mean and cumulative 
distribution function of this distribution in terms of the means and cumulative distribution functions of the mutational fitness distribution, $\psi$, and the corresponding 
equilibrium distribution, $\pi$. Additional results describing how changing the initial fitness of the population changes the distribution of fitnesses after one substitution can 
be found in~\ref{sec:additional}.

More formally, the distribution of fitnesses after one substitution has a probability density function given by $Q(y|x)/k(x)\,dy$, where $Q(y|x)\,dy$ is the substitution rate 
from $x$ to the interval $[y,y+dy]$ and $k(x)=\int_{-\infty}^{\infty} Q(y|x)\,dy$ is the substitution rate at $x$. In order to understand the structure of this distribution, we first state a very general result 
concerning the expected value of any non-decreasing function, $f$, of fitness with respect to the distribution of fitnesses after one substitution as compared with the 
expected value of $f$ for fitnesses drawn from $\psi$ and $\pi$. Deferring the proof of the main result until the end of this Appendix, we then state several corollaries 
corresponding to statements made in the main text, where these corollaries are derived by making specific choices for the non-decreasing function $f$.

Here is the general result:

\begin{theorem}
\label{th:main}
Under the house of cards model with a continuous mutational distribution $\psi$ and with $N>1$, for any arbitrary initial fitness $x$ and non-decreasing, real-valued function 
$f$ defined on the real line, we have
\begin{equation}
\int_{-\infty}^{\infty} f(y)\,\psi(y)\,dy \leq  \int_{-\infty}^{\infty} f(y)\,\frac{Q(y|x)}{k(x)}\,dy \leq \int_{-\infty}^{\infty} f(y)\,\pi(y)\,dy
\end{equation}
whenever the corresponding integrals are finite. The above inequalities are strict if there exist two open intervals of non-zero length within the support of $\psi$ such that the value of $f$ is 
strictly greater for members of the first interval than it is for members of the second.
\end{theorem}

Our biological results then follow by picking particular non-decreasing functions $f$. 

\begin{corollary}[\textbf{Expected fitness after one substitution}]
\label{cor:expfit}
For any $N>1$, continuous probability distribution $\psi$, and any choice of initial fitness:
\begin{enumerate}
\item The expected Malthusian fitness after one substitution is strictly less than the expected Malthusian fitness at equilibrium and strictly greater than the expected 
Malthusian fitness of new mutations, should these expectations be finite.
\item The expected Wrightean fitness after one substitution is strictly less than the expected Wrightean fitness at equilibrium and strictly greater than the expected 
Wrightean fitness of new mutations, should these expectations be finite.
\end{enumerate}
\end{corollary}
\begin{proof}
The first statement follows directly from Theorem~\ref{th:main} by choosing $f(y)=y$. This function is strictly increasing and so the condition for the strictness of the 
inequalities is satisfied. The second statement follows directly from Theorem~\ref{th:main} by choosing $f(y)=e^{y}$ (Malthusian fitness is simply $\log$ Wrightean fitness); 
this function is also strictly increasing and so the inequalities are strict.
\end{proof} 
An immediate consequence of this corollary is that for a population that starts at the equilibrium mean fitness, the first substitution is expected to decrease fitness and so 
the expected selection coefficient of the first substitution is negative. The fact that the expected selection coefficient of the first substitution is a continuous function of the initial fitness (see~\ref{sec:additional} for a proof) implies the existence of some critical initial fitness, strictly less than the equilibrium mean fitness, such that the expected selection 
coefficient of the first substitution is always negative for initial fitnesses greater than this initial fitness, as stated in the main text. See also~\ref{sec:additional} for a proof that the expected selection coefficient of the first substitution is finite whenever either of 1) the mean fitness of new mutations or 2) the mean fitness at equilibrium is finite.

\begin{corollary}[\textbf{Cumulative distribution function of fitness after one substitution}]
\label{cor:cdf}
For $N>1$, any continuous probability distribution $\psi$, any choice of initial fitness $x$, and for all fitnesses $z$ such that the probability that a new mutant has fitness 
less than or equal to $z$ is strictly between $0$ and $1$, the cumulative distribution function of the distribution of fitnesses after one mutation is strictly less than the 
cumulative distribution function of the mutational distribution and, if an equilibrium distribution exists, strictly greater than the cumulative distribution function of the distribution 
of fitnesses at equilibrium.
\end{corollary}
\begin{proof}
Let $f$ be the function that takes the value $1$ on the interval $(z,\infty)$ and 0 otherwise. Then Theorem~\ref{th:main} gives us:
\begin{equation}
\int_{z}^{\infty}\psi(y)\,dy<\int_{z}^{\infty}\frac{Q(y|x)}{k(x)}dy<\int_{z}^{\infty}\pi(y)\,dy
\end{equation}
where the inequalities are strict since $0<\int_{-\infty}^{z}\psi(y)dy<1$ implies that there exists an open interval where $\psi(y)>0$ with $y<z$ and another open interval where 
$\psi(y)>0$ with $y>z$. However, since probability distributions must integrate to $1$, we also have
\begin{equation}
\int_{-\infty}^{z}\psi(y)\,dy>\int_{-\infty}^{z}\frac{Q(y|x)}{k(x)}dy>\int_{-\infty}^{z}\pi(y)\,dy
\end{equation}
as required.
\end{proof}
This corollary then allows us to prove the following result:
\begin{corollary}[\textbf{Quantiles of fitness after one substitution}]
\label{cor:q}
 $N>1$, any continuous probability distribution $\psi$, any choice of initial fitness $x$, and any choice of $q\in(0,1)$, the $q$-th quantile of the distribution of fitnesses after 
one substitution is strictly greater than the $q$-th quantile of the distribution of fitnesses introduced by mutation and, if an equilibrium distribution exists, strictly less than 
the $q$-th quantile of the equilibrium fitness distribution.
\end{corollary}
\begin{proof}
This follows immediately from Corollary~\ref{cor:cdf} by choosing $z$ to be the $q$-th quantile of the distribution of fitnesses after one substitution.
\end{proof}
Results about when a majority of substitutions are deleterious then follow from choosing $q=1/2$ and asking when this quantile is less than the initial fitness. If the initial 
fitness is the equilibrium median fitness, then the median fitness after one substitution is less than the initial fitness, meaning that the probability that the first substitution is 
deleterious is greater than $1/2$. The fact that the probability that the first substitution is deleterious is continuous in the initial fitness (see~\ref{sec:additional} for a proof) then implies 
the existence of the critical initial fitness referred to in the main text.

Corollary~\ref{cor:q} can also provide insight into the dynamics of a population at equilibrium. In particular, note that if the probability that the next substitution is deleterious 
was $q$ for each quantile $q$ of the equilibrium distribution, then the probability that the next substitution would be deleterious for a population whose fitness was drawn at 
random from the equilibrium distribution would be $1/2$. However, Corollary~\ref{cor:q} implies that the probability that the next substitution is deleterious at a particular 
fitness is always strictly greater than the quantile of that fitness in the equilibrium distribution. Thus, the probability that the next substitution is deleterious for a population 
whose fitness is drawn from the equilibrium distribution must be strictly greater than $1/2$.

It is worth noting that the above results can all be extended to the case where the house of cards distribution is discrete or contains point masses (i.e.~where $\psi(y)\,dy$ 
is an arbitrary probability measure instead of a probability density function and expectation is interpreted in terms of Lebesgue integration). However, the statement of the 
results in this more general setting becomes more complex due to the non-zero probability of substitutions that are precisely neutral. See also~\ref{sec:additional} for additional 
proofs relating the distribution of fitnesses after one substitution for a population currently fixed at fitness $x$ to the corresponding distribution for a population currently 
fixed at some other fitness $x'$ and a characterization of the behavior of the distribution of fitnesses after one substitution in the limit where the initial fitness $x$ tends to 
either $\infty$ or $-\infty$.

We now provide the proof for Theorem~\ref{th:main}.
\begin{proof}

Consider a non-decreasing function $f$, and let
\begin{align}
E_x(f) &= \int_{-\infty}^{\infty} f(y)\,\frac{Q(y|x)}{k(x)}\, dy\\
E_{\pi}(f)& = \int_{-\infty}^{\infty} f(y)\,\pi(y)\, dy
\end{align}
be finite. The first thing we want to know is the sign of
\begin{equation}
E_{x}(f) - E_{\pi}(f)= \int_{-\infty}^{\infty} \left( f(y) - E_{\pi}(f)\right) \,\frac{Q(y|x)}{k(x)}\,dy,
\end{equation}
where we have used the fact that $Q(y|x)/k(x)\,dy$ is a probability density and therefore integrates to one. Since $k(x)>0$ it thus suffices to determine the sign of
\begin{equation}
\int_{-\infty}^{\infty} \left( f(y) - E_{\pi}(f)\right) \,Q(y|x)\,dy \mbox{.}
\end{equation}
Now, define $z$ to be any fitness such that $f(y)\geq E_{\pi}(f)$ for $y>z$ and $f(y)\leq E_{\pi}(f)$ for $y<z$. Such a fitness surely exists since $f$ is non-decreasing and $
\min f \leq E_{\pi}(f) \leq \max f$. We can then write
\begin{multline}
\int_{-\infty}^{\infty} \left( f(y) - E_{\pi}(f)\right) \,Q(y|x)\,dy\\=\int_{-\infty}^{z} \left( f(y) - E_{\pi}(f)\right) \,N\,\psi(y)\,u(y-x,N)\,dy\\+\int_{z}^{\infty} \left( f(y) - E_{\pi}(f)\right) \,N
\,\psi(y)\,u(y-x,N)\,dy \label{eq:setup}
\end{multline}
where the first integral on the right hand side is non-positive and the second integral is non-negative. Now, \citet{McCandlish13c} have proved that for $N>1$ and $s>s'$ we have
\begin{equation}
\label{eq:discrim1}
\frac{u(s,N)}{u(s',N)}<e^{(N-1)(s-s')}
\end{equation}
and the inequality is reversed for $s<s'$. Thus, for any three fitnesses, $x$, $y$ and $z$ we have
\begin{equation}
\label{seq:probfixineq}
u(y-x,N)
\begin{cases}
<u(z-x,N)\,e^{(N-1)(y-z)} & \mbox{for } y>z\\
=u(z-x,N)\,e^{(N-1)(y-z)} & \mbox{for } y=z\\
>u(z-x,N)\,e^{(N-1)(y-z)} & \mbox{for } y<z\mbox{.}
\end{cases}
\end{equation}
Using this inequality together with Equation~\ref{eq:setup} yields:
\begin{equation}
\label{eq:secondineq}
\begin{split}
\int_{-\infty}^{\infty} \left( f(y) - E_{\pi}(f)\right) \,Q(y|x)\,dy & \leq \int_{-\infty}^{z} \left( f(y) - E_{\pi}(f)\right) \,N\, \psi(y)\, u(z-x,N)\,e^{(N-1)(y-z)}\,dy \\
& \quad +\int_{z}^{\infty} \left( f(y) - E_{\pi}(f)\right)  \,N\, \psi(y)\, u(z-x,N)\,e^{(N-1)(y-z)}\,dy \\
& \leq N \, u(z-x,N)\, e^{-(N-1)z}\int_{-\infty}^{\infty}  \left( f(y) - E_{\pi}(f)\right) \,\psi(y) \,e^{(N-1)y}\,dy \\
& \leq NZ\, u(z-x,N)\, e^{-(N-1)z} \int_{-\infty}^{\infty}  \left( f(y) - E_{\pi}(f)\right) \,\pi(y)\,dy \\
& \leq NZ\, u(z-x,N)\, e^{-(N-1)z}  \left(\int_{-\infty}^{\infty}  f(y)\,\pi(y)\,dy - E_{\pi}(f)\right) \\
& = 0
\end{split}
\end{equation}
where we have also used the definition of $\pi(y)$, the fact that $\pi$ is probability distribution and hence integrates to $1$, and the definition of $E_{\pi}(f)$. This 
demonstrates the second inequality in the statement of the theorem, $\int_{-\infty}^{\infty} f(y)\,\left(Q(y|x)/k(x)\right)\,dy \leq \int_{-\infty}^{\infty} f(y)\,\pi(y)\,dy$.

To see that this inequality is strict if there exist two open intervals of non-zero length within the support of $\psi$ such that $f$ is strictly greater on one of these intervals than it is on the 
other, notice that in the derivation of the first line of Equation~\ref{eq:secondineq} the inequality for the probability of fixation given by Equation~\ref{seq:probfixineq} is 
strict. Thus, the first line of Equation~\ref{eq:secondineq} can only be an equality if both $\int_{-\infty}^{z} \left( f(y) - E_{\pi}(f)\right) \, \psi(y) \,dy=0$ and $\int_{z}^{\infty} 
\left( f(y) - E_{\pi}(f)\right) \, \psi(y) \,dy=0$. However if there exist two open intervals within the support of $\psi$ such that $f$ is strictly greater on one of these intervals 
than it is on the other, then there exists at least one of 1) some open interval $\alpha=(a,b)$ within the support of $\psi$ with $a<b\leq z$ such that $(f(y)-E_{\pi}(f))<0$ for $y\in (a,b)$  or 2) some open 
interval $\beta=(c,d)$ within the support of $\psi$ with $d>c\geq z$ such that $(f(y)-E_{\pi}(f))>0$ for $y\in (c,d)$. Now, $(f(y)-E_{\pi}(f))\leq0$ for $y\leq z$, so that the existence of $\alpha$ implies that $
\int_{-\infty}^{z} \left( f(y) - E_{\pi}(f)\right) \, \psi(y) \,dy<0$. Similarly, $(f(y)-E_{\pi}(f))\geq0$ for $y\geq z$ so that the existence of $\beta$ implies that $\int_{z}^{\infty} 
\left( f(y) - E_{\pi}(f)\right) \, \psi(y) \,dy>0$. Thus in either case the inequality is strict.

To prove the other inequality in the statement of the theorem, we proceed in a similar manner. Let
\begin{equation}
E_{\psi}(f) = \int_{-\infty}^{\infty} f(y)\,\psi(y)\, dy
\end{equation}
be finite.

Now, we want to know the sign of $E_{x}(f)\,dy - E_{\psi}(f)\,dy$, which takes the same sign as 
\begin{equation}
\int_{-\infty}^{\infty} \left( f(y) - E_{\psi}(f)\right) \,Q(y|x)\,dy \mbox{.}
\end{equation}
Choose $z$ to be a fitness such that $f(y)\geq E_{\psi}(f)$ for $y>z$ and $f(y)\leq E_{\psi}(f)$ for $y<z$. Then
\begin{equation}
\label{eq:firstineq}
\begin{split}
\int_{-\infty}^{\infty} \left( f(y) - E_{\psi}(f)\right) \,Q(y|x)\,dy &= \int_{-\infty}^{z} \left( f(y)-E_{\psi}(f) \right) N\, \psi(y) \,u(y-x,N)\, dy \\
 & \quad + \int_{z}^{\infty} \left( f(y)-E_{\psi}(f) \right) N\, \psi(y) \,u(y-x,N)\, dy \\
 & \geq \int_{-\infty}^{\infty} \left( f(y)-E_{\psi}(f) \right) N\, \psi(y) \,u(z-x,N)\, dy \\
 & = 0,
\end{split}
\end{equation}
where we have used the fact that $u(z-x,N) >u(y-x,N)$ for $y<z$ and $u(z-x,N)<u(y-x,N)$ for $y>z$. This proves the first inequality in the statement of the theorem. 

The argument for the strictness of this inequality if there exist two open intervals within the support of $\psi$ such that $f$ is strictly greater on one of these intervals than 
it is on the other is similar to the argument used for the other inequality and hinges on the strictness of the inequalities $u(z-x,N) >u(y-x,N)$ for $y<z$ and $u(z-x,N)<u(y-x,N)
$ for $y>z$ used to derive the second line in Equation~\ref{eq:firstineq} rather than on the strictness of the inequalities in Equation~\ref{seq:probfixineq}.
\end{proof}

\end{bibunit}

\clearpage
\pagebreak

\setcounter{equation}{0}
\setcounter{theorem}{0}

\begin{bibunit}[evolution]

\section*{Supporting Information}

\setcounter{section}{0}

\appendix
\def\thesection{Supporting Information S\arabic{section}}
\renewcommand{\theequation}{S\arabic{equation}}
\renewcommand{\thetheorem}{S\arabic{theorem}}

\section{Additional results on the distribution of fitnesses after one substitution}
\label{sec:additional}

This section of the Supporting Information contains additional results to complement those derived in~\ref{sec:proofmain}. First, we derive several results related to how the distribution of fitnesses after one substitution changes as the initial fitness is varied. Next, we derive conditions for when the expected value of a non-decreasing function is finite with respect to the distribution of fitnesses after one substitution. Finally, we prove that the substitution rate ($k(x)$), expected selection coefficient of the first substitution ($m(x)$), and proportion of the time the first substitution is advantageous ($p(x)$) all change continuously as a function of initial fitness ($x$).

The following result shows that the distribution of fitnesses after one substitution for a population fixed for fitness $x$ is to the left of the distribution of fitnesses after one 
substitution for a population fixed at fitness $x'>x$. The method of proof is similar to the proof of Theorem~\ref{th:main} in~\ref{sec:proofmain} but uses a different 
inequality for the probability of fixation.

\begin{proposition}
Under the house of cards model with a continuous mutational distribution $\psi$  and $N>1$, for any arbitrary initial fitnesses $x$ and $x'$ with $x'>x$ 
and any non-decreasing, real-valued function $f$ defined on the real line, we have
\begin{equation}
 \int_{-\infty}^{\infty} f(y)\,\frac{Q(y|x)}{k(x)}\,dy \leq  \int_{-\infty}^{\infty} f(y)\,\frac{Q(y|x')}{k(x')}\,dy
\end{equation}
whenever the corresponding integrals exist. The above inequality is strict if there exist two open intervals of length greater than zero within the support of $\psi$ such that the value of $f$ is strictly 
greater for members of the first interval than it is for members of the second.
\end{proposition}
\begin{proof}
Let
\begin{align}
E_x(f) &= \int_{-\infty}^{\infty} f(y)\,\frac{Q(y|x)}{k(x)}\, dy<\infty\\
E_{x'}(f) &= \int_{-\infty}^{\infty} f(y)\,\frac{Q(y|x')}{k(x')}\, dy<\infty.
\end{align}
Also, let $z$ be any fitness such that $f(y) \leq E_{x'}(f)$ for $y<z$ and $f(y) \geq E_{x'}(f)$ for $y>z$. Such a fitness surely exists since $f$ is increasing and $ \min{f}\leq 
E_{x'}(f)\leq \max{f}$.

We are interested in the sign of 
\begin{equation}
E_x(f)-E_{x'}(f)=\int_{-\infty}^{\infty}\left( f(y)-E_{x'}(f)\right)\,\frac{Q(y|x)}{k(x)}\,dy.
\end{equation}
First, we expand
\begin{equation}
\label{eq:suppsplit}
\begin{split}
\int_{-\infty}^{\infty}\left( f(y)-E_{x'}(f)\right)\,\frac{Q(y|x)}{k(x)}\,dy &=\int_{-\infty}^{z}\left( f(y)-E_{x'}(f)\right)\frac{N\, \psi(y)\, u(y-x,N)}{k(x)}\, dy\\
& \quad +\int_{z}^{\infty}\left( f(y)-E_{x'}(f)\right) \frac{N\, \psi(y)\, u(y-x,N)}{k(x)}\, dy,
\end{split}
\end{equation}
where the first term on the right-hand side is non-positive and the second term in non-negative. Now~\citet{McCandlish13c} have shown that for $N>1$, $c>0$ and $s'>s$ 
\begin{equation}
\label{eq:discrim2}
\frac{u(s'+c,N)}{u(s+c,N)}<\frac{u(s',N)}{u(s,N)},
\end{equation}
and that the inequality is reversed for $s'<s$. Thus, for any four 
fitnesses $x$, $x'$, $y$ and $z$ with $x<x'$:
\begin{equation}
\label{seq:probfix4}
\frac{u(y-x,N)}{u(z-x,N)}
\begin{cases}
<\frac{u(y-x',N)}{u(z-x',N)} & \mbox{for } y>z\\
=\frac{u(y-x',N)}{u(z-x',N)} & \mbox{for } y=z\\
>\frac{u(y-x',N)}{u(z-x',N)} & \mbox{for } y<z\mbox{.}
\end{cases}
\end{equation}
Combining this inequality with Equation~\ref{eq:suppsplit} yields:
\begin{equation}
\begin{split}
\int_{-\infty}^{\infty}\left( f(y)-E_{x'}(f)\right)\,\frac{Q(y|x)}{k(x)}\,dy & \leq \int_{-\infty}^{\infty}\left( f(y)-E_{x'}(f)\right)\,N\,\psi(y)\frac{u(z-x,N)}{k(x)}\frac{u(y-x',N)}{u(z-x',N)} \,dy 
\\
 & \leq \frac{u(z-x,N)}{u(z-x',N)} \frac{k(x')}{k(x)}\int_{-\infty}^{\infty}\left( f(y)-E_{x'}(f)\right) N\,\psi(y) \frac{u(y-x',N)}{k(x')}\,dy \\
 & \leq \frac{u(z-x,N)}{u(z-x',N)} \frac{k(x')}{k(x)}\int_{-\infty}^{\infty}\left( f(y)-E_{x'}(f)\right) \frac{Q(y|x')}{k(x')}\,dy \\
 & = 0.
\end{split}
\end{equation}
This demonstrates the inequality in the statement of the proposition. The justification for the condition on the strictness of the inequality is similar to that given in the proof 
of  Theorem~\ref{th:main} except that it relies on the strictness of the inequalities in Equation~\ref{seq:probfix4}.
\end{proof}
This proposition can then be used to show that the expected fitness after one substitution is increasing as a function of the initial fitness, and that, indeed, every quantile of the distribution of fitnesses after one substitution (except for the 0-th and 1-st) is increasing as a function of initial fitness. The proofs of these statements are exactly analogous to the proofs for Corollaries~\ref{cor:expfit},~\ref{cor:cdf} and~\ref{cor:q} given in~\ref{sec:proofmain}.

We also wish to characterize the limit of the distribution of fitnesses after one substitution as the initial fitness becomes extremely large or extremely small. For technical 
reasons, we will conduct this proof in somewhat more generality than the previous proofs in that $\psi(y)\,dy$ will be an arbitrary probability measure and all associated 
integrals will be Lebesgue integrals. Note that in this more general context it is still the case that an equilibrium measure
\begin{equation}
\pi(y) \,dy=\frac{\psi(y)e^{(N-1)y}}{\int_{-\infty}^{\infty}\psi(z)e^{(N-1)z}\, dz} \,dy
\end{equation}
exists for the stochastic process defined by Equation~\ref{eq:master} provided that the normalizing constant $Z=\int_{-\infty}^{\infty}\psi(z)e^{(N-1)z}\, dz$ is finite. We can 
now state the result:

\begin{proposition}
For $N\geq1$, we have
\begin{enumerate}
\item For all fitnesses $y$
\begin{equation}
\lim_{x\rightarrow -\infty} \frac{Q(y|x)}{k(x)} \,dy=\psi(y) \,dy.
\end{equation}
\item If $\pi$ exists, for all fitnesses $y$
\begin{equation}
\lim_{x\rightarrow \infty} \frac{Q(y|x)}{k(x)} \,dy=\pi(y)\,dy.
\end{equation}
\end{enumerate}
\end{proposition}
\begin{proof}
For the first statement in the theorem, we begin by characterizing
\begin{equation}
k(x)=\int_{-\infty}^{\infty}Q(y|x)\,dy=\int_{-\infty}^{\infty}N\, \psi(y)\, u(y-x,N) \,dy.
\end{equation}
Now, $0<u(y-x)\leq 1$ and $\int_{-\infty}^{\infty} \psi(y)\,dy=1<\infty$, so that applying the Lebesgue dominated convergence theorem to move the limit inside the integral 
gives us
\begin{equation}
\begin{split}
\lim_{x\rightarrow -\infty} k(x) &=\int_{-\infty}^{\infty} N\,\psi(y) \lim_{x\rightarrow -\infty} u(y-x,N) \,dy \\
 &=N.
 \end{split}
\end{equation}
where for the second equality have used the fact that $\lim_{x\rightarrow-\infty}u(y-x,N)=1$. Thus
\begin{equation}
\begin{split}
\lim_{x\rightarrow -\infty} \frac{Q(y|x)}{k(x)} \,dy &= \lim_{x\rightarrow -\infty} \frac{N\,\psi(y) u(y-x,N)}{k(x)} dy \\
 &=\lim_{x\rightarrow -\infty} \psi(y) u(y-x,N) dy \\
 &= \psi(y).
\end{split}
\end{equation}
This completes the proof of the first portion of the proposition

To address the second statement in the proposition, where $x\rightarrow \infty$, first observe that
\begin{equation}
\label{eq:rewriterate}
Q(y|x)\, dy = e^{-(N-1)x}\,N\, \psi(y)\, e^{(N-1)y}\, u(x-y,N)\, dy,
\end{equation}
where we have used the identity $u(y-x,N)=e^{(N-1)(y-x)}u(x-y,N)$. Then
\begin{equation}
\begin{split}
\lim_{x\rightarrow \infty}\, k(x)\, e^{(N-1)x} &= \lim_{x\rightarrow \infty}e^{(N-1)x} \int_{-\infty}^{\infty}Q(y|x)\, dy \\
 &= N\,\lim_{x\rightarrow \infty} \int_{-\infty}^{\infty} \psi(y)\, e^{(N-1)y}\, u(x-y,N)\, dy,
\end{split}
\end{equation}
provided that the integral in the last line is finite, which it is since $0<u(y-x,N)\leq1$ and $\int_{-\infty}^{\infty} \psi(y)\,e^{(N-1)y}\,dy=Z<\infty$ by assumption. The two last facts 
also allow us to apply the Lebesgue dominated convergence theorem to move the limit inside the integral, which together with the fact that $\lim_{x\rightarrow 
\infty} u(x-y,N)=1$ yields
\begin{equation}
\lim_{x\rightarrow \infty}k(x)\,e^{(N-1)x}=N\, Z.
\end{equation}
Using this last result and equation~\ref{eq:rewriterate}, we then have
\begin{equation}
\begin{split}
\lim_{x\rightarrow \infty} \frac{Q(y|x)}{k(x)} dy &= \lim_{x\rightarrow \infty} \frac{N\, \psi(y)\, e^{(N-1)y}\,u(x-y,N)}{k(x)\,e^{(N-1)x}}\,dy\\
 &=\frac{N\,\psi(y) e^{(N-1)y}}{N\,Z} \,dy\\
 &= \pi(y)\,dy,
\end{split}
\end{equation}
as required.
\end{proof}

Because many of our results describe inequalities for the expected value of a non-decreasing function $f$ with respect to $\psi$, $\pi$, and the distribution of fitnesses after one substitution, it is helpful to characterize when these expectations exist. The following result shows that if this expectation is finite with respect to either $\psi$ or $\pi$, then it is also finite with respect to the distribution of fitnesses after one substitution:

\begin{proposition}
Under the house of cards model with a continuous mutational distribution $\psi$, $N>1$, and any non-
decreasing, real-valued function $f$ defined on the real line, we have
\begin{enumerate}
\item If $\int_{-\infty}^{\infty} f(y)\,\frac{Q(y|x)}{k(x)}\,dy$ is finite for any $x$, then it is finite for all $x$.
\item If $\int_{-\infty}^{\infty} f(y)\,\psi(y)\,dy$ is finite, then $\int_{-\infty}^{\infty} f(y)\,\frac{Q(y|x)}{k(x)}\,dy$ is finite for all $x$.
\item If $\int_{-\infty}^{\infty} f(y)\,\pi(y)\,dy$ is finite, then $\int_{-\infty}^{\infty} f(y)\,\frac{Q(y|x)}{k(x)}\,dy$ is finite for all $x$.
\end{enumerate}
\end{proposition}
\begin{proof}
If $f$ is constant on the interior of the support of $\psi$, then we are done, since then all the above integrals take this constant values. If $f$ is not constant on the interior of the support of $\psi$, without loss of generality we can assume that $f$ takes both positive and negative values on the interior of the support of $\psi$ since otherwise we can simply consider a function $f^{*}(y)=f(y)-f(w)$ for some $w$ such that $w$ is in the interior of the support of $\psi$, which alters the values of the above integrals by $f(w)$ but does not affect whether the integrals are finite.

Let $z$ be a constant such that $f(y)\geq0$ for $y\geq z$ and $f(y)\leq 0 $ for $y\leq z$. To demonstrate the first claim, note that because $0<k(x)<N$ for all $x$, it suffices to show that $\int_{-\infty}^{\infty}f(y)\,Q(y|x')\,dy$ is finite given that $\int_{-\infty}^{\infty}f(y)\,Q(y|x)\,dy$ is finite for some $x$. Because we will demonstrate the finiteness of these integrals using the limit comparison test, we must make a separate argument for each limit of these improper integrals. We thus consider  \linebreak
$\int_{z}^{\infty}f(y)\,Q(y|x')\,dy$ and $\int_{-\infty}^{z}f(y)\,Q(y|x')\,dy$ separately.

Now, $\lim_{s\rightarrow\infty} u(s,N)=1$ and $\lim_{s\rightarrow-\infty} u(s,N)\,e^{(N-1)s}=1$ so that
\begin{equation}
\lim_{y\rightarrow \infty} \frac{u(y-x',N)}{u(y-x,N)}=1
\end{equation}
and
\begin{equation}
\begin{split}
\lim_{y\rightarrow -\infty} \frac{u(y-x',N)}{u(y-x,N)} &= \lim_{y\rightarrow -\infty} \frac{u(y-x',N)\,e^{(N-1)(y-x')}}{u(y-x,N)\,e^{(N-1)(y-x)}}\,e^{(N-1)(x'-x)} \\
&=e^{(N-1)(x'-x)}.
\end{split}
\end{equation}
Thus, we have
\begin{equation}
\lim_{y\rightarrow \infty} \frac{f(y)\,Q(y|x')}{f(y)\,Q(y|x)}=1>0,
\end{equation}
which shows that $\int_{z}^{\infty}f(y)\,Q(y|x')\,dy$ is finite. Similarly, we have 
\begin{equation}
\lim_{y\rightarrow -\infty} \frac{f(y)\,Q(y|x')}{f(y)\,Q(y|x)}=e^{(N-1)(x'-x)}>0,
\end{equation}
which shows that $\int_{-\infty}^{z}f(y)\,Q(y|x')\,dy$ is finite.

To show the second claim, note that because $0<k(x)<N$ for all $x$, it suffices to show that both $\int_{-\infty}^{z} f(y)\,\psi(y) \,u(y-x,N)\,dy$ and $\int_{z}^{\infty} f(y)\,\psi(y) \,u(y-x,N)\,dy$ are finite. This result follows because $0<u(y-x,N)<1$ and we are given that $\int_{-\infty}^{z} f(y)\,\psi(y) \,dy$ and $\int_{z}^{\infty} f(y)\,\psi(y) \,dy$ are finite.

To show the third claim, it suffices to show that  $\int_{z}^{\infty} f(y)\,\psi(y) \,u(y-z,N)\,dy$ and \linebreak $\int_{-\infty}^{z} f(y)\,\psi(y) \,u(y-z,N)\,dy$ are finite, since then $\int_{-\infty}^{\infty} f(y)\,\left(Q(y|x)/k(x) \right)\,dy$ is finite for all $x$ by the first claim. To show that these integrals are finite, we note that $0<u(s,N) < e^{(N-1)s}$ for $N>1$~\citep{McCandlish13c} so that
\begin{equation}
0 \leq \int_{z}^{\infty}f(y)\, \psi(y)\,u(y-z,N)\,dy \leq e^{-(N-1)z} \int_{z}^{\infty}f(y)\, \psi(y)\,e^{(N-1)y}\,dy<\infty
\end{equation}
and
\begin{equation}
0 \geq \int_{-\infty}^{z}f(y)\, \psi(y)\,u(y-z,N)\,dy \geq e^{-(N-1)z} \int_{-\infty}^{z}f(y)\, \psi(y)\,e^{(N-1)y}\,dy>-\infty,
\end{equation}
as required.
\end{proof}

Finally, for \ref{sec:proofmain}, we also need proofs that $m(x)$ and $p(x)$ are continuous functions of $x$. First we need to prove some simpler results. We begin by proving that 
the substitution rate is continuous with respect to initial fitness.
\begin{proposition}
\label{prop:k}
$k(x)$ is continuous in $x$.
\end{proposition}
\begin{proof}
First note that $k(x)$ exists for all $x$ because $0<u(s,N)\leq1$ so that $0<k(x)\leq N$. We want to show that for any $\epsilon>0$, there exists a $\delta>0$ such that if 
$0<|x-x'|<\delta$ then $|k(x)-k(x')|<\epsilon$. In fact, in this case we can choose $\delta< \epsilon/N$. First, we note that $u(s,N)$ is continuous and $0<\frac{d}{ds}u(s,N)
<1$\citep{McCandlish13c}, so that $0>\frac{d}{dx}u(y-x,N)>-1$, and thus $|u(y-x,N)-u(y-x',N)|<|x-x'|$. We then have
\begin{equation}
\begin{split}
|k(x)-k(x')| &= N \int_{-\infty}^{\infty}|u(y-x,N)-u(y-x',N)|\,\psi(y)\,dy \\ 
 &< N |x-x'|\int_{-\infty}^{\infty} \,\psi(y)\,dy \\
 &< \epsilon,
\end{split}
\end{equation}
as required. 
\end{proof}
Note that since $\int_{-\infty}^{\infty} \,\psi(y)\,dy=1$, the above proof has in fact demonstrated that $|k(x)-k(x')|\leq N|x-x'|$ for all $x, x'$. Thus $k(x)$ is not merely continuous, but is rather Lipschitz continuous with Lipschitz constant $N$.

Next, we prove a technical lemma.

\begin{lemma}
\label{lemma:qcont}
If $\psi$ or $\pi$ has a finite mean, then $\int_{-\infty}^{\infty} y\,Q(y|x)\,dy$ is a continuous function of $x$.
\end{lemma}
\begin{proof}
If $\psi$ has finite mean, then $\int_{-\infty}^{\infty}|y|\,\psi(y)\,dy<\infty$. For $\epsilon>0$, choose $0<\delta<\epsilon/\left(N\int_{-\infty}^{\infty}|y|\,\psi(y)\,dy\right)$. 
Noting again that $|u(y-x,N)-u(y-x',N)|\leq|x-x'|$, for $0<|x-x'|<\delta$ we have
\begin{equation}
\begin{split}
\left|\int_{-\infty}^{\infty} y\,Q(y|x)\,dy-\int_{-\infty}^{\infty} y\,Q(y|x')\,dy\right| &\leq N \int_{-\infty}^{\infty}|y|\,|u(y-x,N)-u(y-x',N)|\,\psi(y)\,dy \\ 
 &\leq N |x-x'|\int_{-\infty}^{\infty} \,|y|\,\psi(y)\,dy \\
 &< \epsilon,
\end{split}
\end{equation}
as required.

Now, suppose $\pi$ has finite mean so that $\int_{-\infty}^{\infty}|y|\,\pi(y)\,dy<\infty$. We first note that:
\begin{equation}
\begin{split}
\left|\frac{d}{dx}\,e^{-(N-1)x}\,u(x-y,N)\right|&=\left|-(N-1)\,e^{-(N-1)x}\,u(x-y,N)+e^{-(N-1)x}\,\frac{d}{dx} u(x-y,N)\right| \\
 &\leq \left| \vphantom{\frac{d}{dx}}(N-1)\,e^{-(N-1)x}\,u(x-y,N)\right|+\left| e^{-(N-1)x}\,\frac{d}{dx} u(x-y,N)\right| \\
 &\leq N\,e^{-(N-1)x}
\end{split}
\end{equation}
where we have used the facts that $0<u(x-y,N)\leq1$ and $0\leq\frac{d}{dx} u(x-y,N)\leq1$. Thus, for all $x,x'$, 
\begin{equation}
\left| e^{-(N-1)x}\,u(x-y,N)-e^{-(N-1)x'}\,u(x'-y,N)\right| \leq N\,e^{-(N-1) \min(x,x')}\left|x-x'\right|.
\end{equation}

Now, for any given choice of $x$ and choice of $\epsilon>0$, pick $\delta>0$ such that $\delta\,e^{(N-1)\delta}<\epsilon/\left(e^{-(N-1)x}\,N^{2}\,Z\,\int_{-\infty}^{\infty} |y|\,\pi(y)\,dy \right)$, which is possible 
since $\delta\, e^{(N-1)\,\delta}$ can be made arbitrarily small. Using the identities $u(s,N)=e^{(N-1)s}u(-s,N)$ and $\pi(y)=e^{(N-1)y}\psi(y)/Z$, for $0<|x-x'|<\delta$ we have
\begin{equation}
\begin{split}
\lefteqn{\left|\int_{-\infty}^{\infty} y\,Q(y|x)\,dy-\int_{-\infty}^{\infty} y\,Q(y|x')\,dy\right|} \quad \quad &\\ 
 &= N \left| \int_{-\infty}^{\infty} y\left( u(y-x,N)-u(y-x',N)\right) \psi(y)\,dy\right| \\
 &= N\,Z \left| \int_{-\infty}^{\infty} y\left( u(y-x,N)-u(y-x',N)\right) e^{-(N-1)y}\,\pi(y)\,dy\right| \\ 
 &= N\,Z \left| \int_{-\infty}^{\infty} y\left( e^{-(N-1)x}u(x-y,N)-e^{-(N-1)x'}u(x'-y,N)\right)\pi(y)\,dy\right| \\
 &\leq \left|x-x'\right|\,e^{-(N-1) \min(x,x')} N^{2}\,Z \int_{-\infty}^{\infty} |y|\,\pi(y)\,dy \\
 &\leq \left|x-x'\right| \,e^{(N-1)|x-x'|}\,e^{-(N-1)x}\,N^{2}\,Z \int_{-\infty}^{\infty} |y|\,\pi(y)\,dy \\
 &<\epsilon,
\end{split}
\end{equation}
as required.
\end{proof}

Together, these results allow us to address the continuity of $m(x)$.

\begin{proposition}
If $\psi$ or $\pi$ has finite mean then $m(x)$ is a continuous function of $x$.
\end{proposition}
\begin{proof}
Because
\begin{equation}
m(x)=\frac{\int_{-\infty}^{\infty}y\,Q(y|x)\,dy}{k(x)}-x,
\end{equation}
the continuity of $m(x)$ follows immediately from the continuity of $\int_{-\infty}^{\infty}y\,Q(y|x)\,dy$ (Lemma~\ref{lemma:qcont}), the continuity of $k(x)$ (Proposition~
\ref{prop:k}), the fact that $k(x)$ is strictly positive, and the continuity of $x$.
\end{proof}

Finally, we address the continuity of $p(x)$.

\begin{proposition}
$p(x)$ is continuous in $x$.
\end{proposition}
\begin{proof}
Because
\begin{equation}
p(x)=\frac{\int_{x}^{\infty}Q(y|x)\,dy}{k(x)}
\end{equation}
and $k(x)$ is both continuous (Proposition~\ref{prop:k}) and strictly positive, it suffices to show that $\int_{x}^{\infty}Q(y|x)\,dy$ is continuous in $x$.

Now, for $\epsilon>0$, choose $\delta>0$ such that $\delta<\epsilon/(2N)$ and $\int_{x-\delta}^{x+\delta} \psi(y)\,dy<\epsilon/(2N)$. It is always possible to choose an appropriate $\delta$ to satisfy the second inequality because we have assumed that $\psi(x)$ is a continuous probability distribution, i.e.~it has no point masses. Thus the total probability in an open ball of radius $\delta$ around $x$ can be made arbitrarily small by choosing a sufficiently small $\delta$. 

Assuming that $|x-x'|<\delta$ we have
\begin{equation}
\label{eq:p1}
\begin{split}
\left| \int_{x}^{\infty} Q(y|x)\,dy-\int_{x'}^{\infty} Q(y|x')\,dy \right| &\leq \left| \int_{x}^{\infty} Q(y|x)\,dy-\int_{x}^{\infty} Q(y|x')\,dy \right|\\
 &\quad \quad+\left|\int_{x}^{x'}Q(y|\min(x,x'))\,dy\right|.
\end{split}
\end{equation}
To analyze the first term on the right-hand side, we again use the fact that $|u(y-x,N)-u(y-x',N)|\leq|x-x'|$, so that for $0<|x-x'|<\delta$ we have
\begin{equation}
\label{eq:p2}
\begin{split}
\left| \int_{x}^{\infty} Q(y|x)\,dy-\int_{x}^{\infty} Q(y|x')\,dy \right| &\leq N\, |x-x'|\int_{x}^{\infty}\psi(y)\,dy \\
 &<\epsilon/2.
\end{split}
\end{equation}
To analyze the other term, we note that $0\leq u(s,N)\leq 1$ so that for $0<|x-x'|<\delta$ we have
\begin{equation}
\label{eq:p3}
\begin{split}
\left|\int_{x}^{x'}Q(y|\min(x,x'))\,dy\right| & \leq N \left|\int_{x}^{x'} \psi(y)\,dy \right| \\
& \leq N \left|\int_{x-\delta}^{x+\delta} \psi(y)\,dy \right| \\
 & < \epsilon/2
\end{split}
\end{equation}
Substituting Equation~\ref{eq:p2} and Equation~\ref{eq:p3} into Equation~\ref{eq:p1} yields the desired result.
\end{proof}

\section{The Gaussian HOC}
\label{sec:GuassianProofs}
In the main text, we claim that for the Gaussian HOC,
\begin{equation}
\frac{Q(x^{*}+z|x^{*}+c)}{k(x^{*}+c)}\,dz=\frac{Q(x^{*}-z|x^{*}-c)}{k(x^{*}-c)}\,dz
\end{equation}
for all $c$ and $z$, where $x^{*}=\mu+(N-1)\sigma^2/2$. To demonstrate this equality, it is sufficient to show that $
{Q(x^{*}+z|x^{*}+c)}\propto{Q(x^{*}-z|x^{*}-c)}$ when viewed as a function of $z$ (i.e.~with $c$ and $x^{*}$ held constant), since two probability density functions that are 
proportional to each other must be equal. Noting that for the Gaussian HOC $\pi(x^{*}+z)=\psi(x^{*}-z)$, and that more generally $u(s,N)=u(-s,N)\,e^{(N-1)s}$, and $\pi(y)\propto 
\psi(y)e^{(N-1)y}$, we have
\begin{equation}
\begin{split}
{Q(x^{*}+z|x^{*}+c)} &= N\, \psi(x^{*}+z)\, u((x^{*}+z)-(x^{*}+c),N) \\
 &=N\, \psi(x^{*}+z)\, u(z-c,N) \\
 &=N\, \pi(x^{*}-z)\, u(c-z,N)\, e^{(N-1)(z-c)} \\
 &\propto \psi(x^{*}-z)\, u(c-z,N)\, e^{(N-1)(x^{*}-c)} \\
 &\propto \psi(x^{*}-z)\, u((x^{*}-z)-(x^{*}-c),N)\\
 &\propto {Q(x^{*}-z|x^{*}-c)}
 \end{split}
\end{equation}
as required.

\section{Analytical approximations for the House of Cards model}
\label{sec:analytical}

It is often useful to develop analytical approximations for quantities of evolutionary interest. Here we present a system for developing analytical approximations under the 
House of Cards model. This system has the surprising feature that the terms in the analytical approximations are very closely related to the mutational and equilibrium 
distributions of the underlying House of Cards model.

Our main strategy is to extend the standard approximation for the probability of fixation, $u(s,N)\approx s$, to accommodate deleterious fixations:
\begin{equation}
\label{eq:currentapprox}
u(s,N)\approx \left\{\begin{aligned}  &-s\,e^{(N-1)s} && \text{for } s \leq 0\\
& s && \text{for } s>0 
\end{aligned} \right. 
\end{equation}
\citep{McCandlish13c}. This approximation has the interesting property that when used in place of the true probability of fixation to define the transition rates in Equation~
\ref{eq:rates}, the resulting Markov process has an equilibrium distribution identical to the equilibrium distribution obtained using the true probability of fixation (Equation~
\ref{eq:eq}).

Using this approximation, we can then analyze the dynamics as follows. For a population fixed for an allele with fitness $x$, the rate of beneficial substitutions, $k_{b}(x)$, 
is
\begin{equation}
k_{b}(x)=\int_{x}^{\infty} Q(y|x)\,dy \approx N \int_{x}^{\infty}(y-x)\psi(y)\,dy
\end{equation}
while the rate of deleterious fixations is
\begin{equation}
\begin{split}
k_{d}(x) &=\int_{-\infty}^{x} Q(y|x)\,dy \\
& \approx -N \int_{-\infty}^{x}(y-x)\,e^{(N-1)(y-x)}\,\psi(y)\,dy \\
& \approx - N\,Z\,e^{-(N-1) x}\int_{-\infty}^{x}(y-x)\, \pi(y)\, dy
\end{split}
\end{equation}
Similarly, the expected instantaneous rate of change in fitness due to beneficial fixations is
\begin{equation}
\begin{split}
r_{b}(x) & =\int_{x}^{\infty} (y-x)\,Q(y|x)\,dy \\
 & \approx N \int_{x}^{\infty}(y-x)^{2}\,\psi(y)\,dy
\end{split}
\end{equation} 
while the expected instantaneous rate of change in fitness due to deleterious fixations is
\begin{equation}
\begin{split}
r_{d}(x) &=\int_{-\infty}^{x} (y-x)\,Q(y|x)\,dy \\
& \approx -N \int_{-\infty}^{x}(y-x)^{2}\,e^{(N-1)(y-x)}\,\psi(y)\,dy \\
& \approx - N\,Z\,e^{-(N-1) x}\int_{-\infty}^{x}(y-x)^{2}\, \pi(y)\, dy.
\end{split}
\end{equation}
To put this another way, approximating the substitution rate and the expected change in fitness for a population fixed for fitness $x$ comes down to calculating the first few 
moments of truncated versions of $\psi$ and $\pi$. For many common distributions, these moments are easy to compute or can be looked up in a table~\citep[e.g., Table 
1 of][]{Jawitz04}.

With these few quantities in hand, we can calculate a variety of other quantities of interest. Of course, the total substitution rate is 
\begin{equation}
k(x)=k_{b}(x)+k_{d}(x),
\end{equation}
and the probability that the next substitution is advantageous is
\begin{equation}
p(x)=\frac{k_{b}(x)}{k(x)}.
\end{equation}
Similarly, the expected instantaneous rate of change in fitness is
\begin{equation}
r(x)=r_{b}(x)+r_{d}(x)
\end{equation}
and the expected selection coefficient of the first substitution is
\begin{equation}
m(x)=r(x)/k(x).
\end{equation}
One can also compute the expected selection coefficient of the first substitution conditional on it being advantageous, $r_{b}(x)/k_{b}(x)$, or deleterious, $r_{d}(x)/k_{d}(x)$.

In the case when $\psi$ is normal with mean $\mu$ and variance $\sigma^{2}$, $\pi$ is normal with mean $\mu'=\mu+(N-1)\sigma^{2}$ and the same variance. Let us write 
the cumulative distribution function of $\psi$ and $\pi$ as $\Psi(x)=\int_{-\infty}^{x}\psi(y)\,dy$ and $\Pi(x)=\int_{-\infty}^{x}\pi(y)\,dy$, respectively. Looking up the 
necessary moments~\citep{Jawitz04} and simplifying leads to the following expressions:
\begin{align}
k_{b} &=N\,\psi(x)\left(\sigma^{2}+\left(\mu-x\right)\frac{1-\Psi(x)}{\psi(x)}\right)\\
k_{d} &=N\,\psi(x)\left(\sigma^{2}-\left(\mu'-x\right)\frac{\Pi(x)}{\pi(x)}\right)\\
r_{b} &= N\,\psi(x)\left(\left(\mu-x\right)\sigma^{2}+\left(\mu-x\right)^{2}\frac{1-\Psi(x)}{\psi(x)}\right)\\
r_{d} &= N\,\psi(x)\left(\left(\mu'-x\right)\sigma^{2}-\left(\mu'-x\right)^{2}\frac{\Pi(x)}{\pi(x)}\right).
\end{align}
It is interesting to ask why the above expressions are so symmetric. This is in large part a consequence of the fact that for our approximation of $u(s,N)$, it is still the 
case that $u(s,N)=u(-s,N)e^{(N-1)s}$ and thus the symmetry argument given in~\ref{sec:GuassianProofs} for the distribution of fitnesses after one substitution still holds 
under the approximation. As a consequence, the approximation is exactly correct for $m(x^{*})$ and $k(x^{*})$, where $x^{*}=(\mu+\mu')/2$.

The above equations were used to plot the curves in Figure~\ref{fig:HOCstats}. Note that the approximation, while
reasonably good, is always more extreme than the numerical results, in the sense
that, e.g., the analytical result for $m(x)$ is too high when $m(x)$ is positive
and too low when $m(x)$ is negative. This is to be expected, since our
approximation underestimates the contribution of nearly neutral fixations and therefore overestimates the average magnitude of the mutations that fix.

\section{Results for Fisher's Geometric Model}
\label{sec:FGMproofs}

\subsection{Scaling of the model}
\label{sec:FGMscaling}

In the main text, we note that evolution under Fisher's geometric model is particularly easy to understand when fitnesses are measured in scaled fitness because then, to a 
very good approximation, the dynamics depend only on the dimensionality $n$ and the compound parameter $N\lambda$. To see why this is the case, recall that the 
fitness of a mutant introduced by mutation when the current fitness is $x$ is $-\lambda/2$ times a non-central chi-squared distributed random variable with $n$ degrees of 
freedom and non-centrality parameter $-2x/\lambda$. Thus, if $X=Nx$ is the current scaled fitness, the scaled fitness of a new mutant is $-\frac{N\lambda}{2}$ times a 
non-central chi-squared random variable with $n$ degrees of freedom and non-centrality parameter $-2x/\lambda=-\frac{2}{N \lambda} X$. Thus, the mutational process, 
when viewed in terms of scaled fitness, depends only on $n$ and the product $N\lambda$.

Furthermore, it is well-known that the evolutionary dynamics under weak mutation depends primarily on scaled fitness, and not on fitness and population size separately, 
so long as selection coefficients remain small. In particular, if we approximate the probability of fixation as
\begin{equation}
u(s,N)=\frac{1-e^{-s}}{1-e^{-Ns}}\approx \frac{s}{1-e^{-Ns}},
\end{equation}
then the instantaneous transition rate from scaled fitness $X$ to scaled fitness $Y$ is given by:
\begin{equation}
\label{eq:scaledrate}
Q(Y|X)\approx\Phi_{X}(Y) \frac{Y-X}{1-e^{-(Y-X)}},
\end{equation}
where $\Phi_{X}(Y)$ is the probability density function describing the scaled fitnesses of new mutations for a population whose current scaled fitness is $X$~\citep{Kryazhimskiy09}. Since this probability density function depends only on $n$, the 
compound parameter $N\lambda$ and the current fitness $X$, it is clear that the dynamics of evolution under FGM when viewed in terms of scaled fitnesses depend only 
on $n$ and $N\lambda$ provided that the absolute size of the selection coefficients is sufficiently small.

Finally, using the above approximation for the transition rate from scaled fitness $X$ to scaled fitness $Y$ yields a simple form for the equilibrium distribution. In particular 
the equilibrium distribution in terms of scaled fitnesses is $-1$ times a gamma distribution with shape $n/2$ and scale 1. Note that this equilibrium distribution is independent of $N
\lambda$ and that its mean is simply $-n/2$.

\subsection{The large $N\lambda$ limit}
\label{sec:largelim}
Having shown that the dynamics of evolution under FGM, to a good approximation, depend only on $n$ and $N\lambda$ when measured in terms of scaled fitnesses, it is 
natural to ask what happens in various limiting cases of these two parameters. Here we will discuss the dynamics as $N\lambda$ becomes large. In particular, we will take 
an approach based on approximating the distribution of scaled fitnesses introduced by mutation, $\Phi_{X}$, under the assumption that terms of order $1/(N\lambda)$ are 
negligible.

In order to investigate the large $N\lambda$ limit, it is helpful to write out the PDF of $\Phi_{X}(Y)$ explicitly. Recalling that $X,Y\leq0$, we have:
\begin{equation}
\Phi_{X}(Y)=\frac{1}{N\lambda}e^{\frac{Y}{N\lambda}+\frac{X}{N\lambda}} \left(\frac{Y}{X}\right)^{\frac{1}{2}(n/2-1)}I_{n/2-1}\left(\sqrt{\left(2\frac{Y}{N\lambda}\right)
\left(2\frac{X}{N\lambda}\right)}\right)
\end{equation}
where $I_v(x)$ is a modified Bessel function of the first kind:
\begin{equation}
I_v(x)=\sum_{k=0}^{\infty} \frac{1}{k!\, \Gamma(k+v+1)}\left(\frac{x}{2}\right)^{2k+v}.
\end{equation} 
Following~\citet{Martin08} in using the method of~\citet{Jaschke04}, we then note that
\begin{gather}
e^{\frac{Y}{N\lambda} } = 1+\frac{Y}{N\lambda}+O\left(\left(\frac{1}{N\lambda}\right)^{2}\right) \\
I_{n/2-1}\left(\sqrt{\left(2\frac{Y}{N\lambda}\right)\left(2\frac{X}{N\lambda}\right)}\right)=\left(\frac{1}{N\lambda}\right)^{n/2-1}\left(YX\right)^{(1/2)(n/2-1)}\left(\frac{1}
{\Gamma(n/2)}+O\left(\left(\frac{1}{N\lambda}\right)^{2}\right)\right)
\end{gather}
so that
\begin{equation}
\Phi_{X}(Y)=e^{X/({N\lambda})}\frac{(-Y)^{n/2-1}}{\Gamma(n/2) (N\lambda)^{n/2}}\left(1+O(\frac{1}{N\lambda})\right).
\end{equation}
We can thus make the approximation
\begin{equation}
\Phi_{X}(Y)\approx e^{X/({N\lambda})}\frac{(-Y)^{n/2-1}}{\Gamma(n/2) (N\lambda)^{n/2}}
\end{equation}
for large $N\lambda$. Note that this is identical to the tail approximation for $x$ close to 0 used by~\citet{Martin08} when translated into scaled fitness and adapted to the 
case of an isotropic FGM as analyzed here. This equivalence makes sense if we consider what happens if we make $N\lambda$ large by increasing $N$ while keeping $
\lambda$ and $X$ fixed, since then $x=X/N$ approaches 0 just as in the approximation developed by~\citet{Martin08}.

\subsection{Analytical approximations in the large $N\lambda$ limit}
\label{sec:FGManalytic}

Given the approximation to the evolutionary dynamics for large $N\lambda$ developed in the previous section (\ref{sec:largelim}), one can develop analytical results for the 
probability that the next substitution is advantageous and the expected scaled selection coefficient of the next substitution by using a further approximation for the 
probability of fixation. In particular, we will approximate the term:
\begin{equation}
\frac{Y-X}{1-e^{-(Y-X)}}
\end{equation}
in Equation~\ref{eq:scaledrate} as
\begin{equation}
\label{eq:currentapprox2}
\frac{Y-X}{1-e^{-Y-X}}\approx \left\{\begin{aligned}  &-(Y-X)\,e^{Y-X} && \text{for } Y-X \leq 0\\
& Y-X && \text{for } Y-X>0.
\end{aligned} \right. 
\end{equation}
This approximation is very similar to the approximation used in \ref{sec:analytical} in that its use results in the same equilibrium distribution as the Markov process given 
by Equation~\ref{eq:scaledrate}. Importantly, this approximation can similarly be viewed as a modification of the standard 
approximation used in the literature on the genetics of adaptation (where the rate of advantageous substitutions is proportional to the selection coefficient) to accommodate the possibility of deleterious substitutions.

To begin, let us recall that for $b>-1$ and $a\leq0$
\begin{equation}
\label{eq:int1}
\int_{a}^{0}(-z)^{b}\,dz=\frac{(-a)^{b+1}}{b+1}
\end{equation}
 and
\begin{equation}
\label{eq:int2}
\int_{-\infty}^{a}e^{z}\,(-z)^{b}\,dz=\Gamma(b+1,-a),
\end{equation}¥
where $\Gamma(b+1,-a)$ is an upper incomplete gamma function.

Using the first of the above integrals (Equation~\ref{eq:int1}), we then have that the rate of advantageous substitution is
\begin{equation}
\label{eq:FGMkb}
\begin{split}
k_{b}(X) &\approx \frac{ e^{X/({N\lambda})}}{\Gamma(n/2) (N\lambda)^{n/2}} \int_{X}^{0}(Y-X)(-Y)^{n/2-1}\,dy\\
 &\approx \frac{ e^{X/({N\lambda})}}{\Gamma(n/2) (N\lambda)^{n/2}} \frac{(-X)^{n/2+1}}{(n/2)(n/2+1)}.
\end{split}
\end{equation}
Using the second of the above integrals (Equation~\ref{eq:int2}), we also have that the rate of deleterious substitution is
\begin{equation}
\label{eq:FGMkd}
\begin{split}
k_{d}(X) & \approx -\frac{ e^{X/({N\lambda})}}{\Gamma(n/2) (N\lambda)^{n/2}} \int_{-\infty}^{X}(Y-X)e^{(Y-X)}(-Y)^{n/2-1}\,dy\\
 &\approx \frac{ e^{X/({N\lambda})}}{\Gamma(n/2) (N\lambda)^{n/2}}\, e^{-X}\left( \Gamma(n/2+1,-X)+X\, \Gamma(n/2,-X)\right) \\
  &\approx \frac{ e^{X/({N\lambda})}}{\Gamma(n/2) (N\lambda)^{n/2}}\, \left((-X)^{n/2}+e^{-X}\,(n/2+X)\,\Gamma(n/2,-X)\right)
\end{split}
\end{equation}
where we have used the recurrence relation $\Gamma(\alpha+1,\beta)=\alpha\, \Gamma(\alpha,\beta)+\beta^{\alpha}\,e^{-\beta}$ to derive the last line. Thus, we have
\begin{equation}
\label{eq:FGMp}
\begin{split}
p(X)&=\frac{k_{b}(X)}{k_{b}(X)+k_{d}(X)}\\
&\approx \frac{-X}{-X+(n/2)(n/2+1)\left(1+e^{-X}\,(n/2+X)\,(-X)^{-(n/2)}\,\Gamma(n/2,-X)\right)}.
\end{split}
\end{equation}
Notice that the probability in the last line is not a function of $N\lambda$. Thus, while the absolute substitution rate at any fixed $X$ decreases for large 
$N\lambda$ as $N\lambda$ increases (see the factor $e^{X/(N\lambda)}/(N\lambda)^{n/2}$ in Equations \ref{eq:FGMkb}~and~\ref{eq:FGMkd}), the fraction of beneficial 
substitutions approaches a definite limit as $N\lambda\rightarrow \infty$.

Now, if we approximate the median fitness at equilibrium by the mean fitness at equilibrium, Equation~\ref{eq:FGMp} shows that the probability that the next substitution is 
advantageous at the equilibrium median fitness is simply
\begin{equation}
p(-n/2)\approx\frac{1}{n/2+2},
\end{equation}
as stated in the main text (as a formal matter, the simplification arises because the term $(n/2+X)$  becomes $0$ if $X=-n/2$). Critically, this expression is always less 
than $1/2$ for $n>0$.

The expected selection coefficient of the first substitution, conditional on it being advantageous, can also be found using Equation~\ref{eq:int1}:
\begin{equation}
\begin{split}
m_{b}(X)&\approx \frac{\int_{X}^{0}Y(Y-X)(-Y)^{n/2-1}\,dy}{\int_{X}^{0}(Y-X)(-Y)^{n/2-1}\,dy}-X \\
 &\approx -X\frac{2}{n/2+2},
\end{split}¥
\end{equation}¥
which is equivalent to the expression given by~\citet{Martin08} when the expression given by~\citet{Martin08} is translated into scaled fitness. At the equilibrium mean 
fitness, this gives us
\begin{equation}
m_{b}(-n/2)\approx \frac{n}{n/2+2}.
\end{equation}

Similarly, the expected selection coefficient of the first substitution, conditional on it being deleterious, can be found using Equation~\ref{eq:int2} and the recursive formula 
for the incomplete gamma function:
\begin{equation}
\begin{split}
m_{d}(X)&\approx \frac{\int_{-\infty}^{X}Y(Y-X)e^{(Y-X)}(-Y)^{n/2-1}\,dy}{\int_{-\infty}^{X}(Y-X)e^{(Y-X)}(-Y)^{n/2-1}\,dy}-X \\
 &\approx \frac{X\,\Gamma(n/2+1,-X)+\Gamma(n/2+2,-X)}{X\,\Gamma(n/2,-X)+\Gamma(n/2+1,-X)}-X \\
 &\approx-\frac{n}{2} \left(1+\frac{\Gamma(n/2+1,-X)}{(n/2+X)\Gamma(n/2+1,-X)+(-X)^{n/2+1}e^{-X} } \right)-X.
\end{split}¥
\end{equation}¥
At the equilibrium mean fitness, this expression simplifies further to
\begin{equation}
m_{d}(-n/2) \approx-\left(\frac{e}{n/2}\right)^{n/2}\Gamma(n/2+1,n/2).
\end{equation}

We can now use these results to characterize the expected selection coefficient of the first substitution for a population whose initial fitness is equal to the equilibrium 
mean fitness. The expected selection coefficient is given by:
\begin{equation}
\begin{split}
m(-n/2) &\approx p(-n/2)\,m_{b}(-n/2)+\left(1-p(-n/2)\right)\,m_{d}(-n/2)\\
 &\approx \frac{n}{(n/2+2)^{2}}-\left(\frac{e}{n/2}\right)^{n/2}\Gamma(n/2+1,n/2)\,\frac{n/2+1}{n/2+2},
\end{split}¥
\end{equation}¥
which is solely a function of $n$. Recalling the well-known inequality $\Gamma(\alpha,\beta)\geq e^{-\beta}\,\beta^{\alpha-1}$ for $\alpha\geq1$ \citep[see, e.g.,][Section 8.10]
{Olver10}, we have
\begin{equation}
\begin{split}
\frac{n}{(n/2+2)^{2}}-\left(\frac{e}{n/2}\right)^{n/2}\Gamma(n/2+1,n/2)\,\frac{n/2+1}{n/2+2} &\leq \frac{n}{(n/2+2)^{2}}-\frac{n/2+1}{n/2+2}\\
 &\leq -\frac{(n/2)^{2}+n/2+2}{n/2+2}\\
 &<0
\end{split}¥
\end{equation}¥
for $n\geq1$, so that our approximation of $m(n/2)$ is always negative.

\renewcommand{\refname}{Literature Cited in Supporting Information}
\putbib[MainBibtexDatabase]
\end{bibunit}

\end{document}